\pdfoutput=1

\documentclass[10pt,pra,twocolumn,superscriptaddress,showpacs,nofootinbib]{revtex4}

\usepackage{graphicx}
\usepackage{amssymb}
\usepackage{amsmath}
\usepackage{color}

\newtheorem{thrm}{Theorem}

\input{Qcircuit.tex}

\def\controlled#1{\mathrm{C}_{#1}}
\def\CZ{\controlled Z}

\newcommand{\reals}{\mathbb{R}}

\newcommand{\avg}[1]{\left\langle{{#1}}\right\rangle}
\newcommand{\norm}[1]{\left\|{{#1}}\right\|}

\DeclareMathOperator{\Det}{Det}

\begin{document}

\title{Quantum Computing with Continuous-Variable Clusters}

\author{Mile Gu}
\affiliation{Department of Physics, University of Queensland, St Lucia, Queensland 4072, Australia.}

\author{Christian Weedbrook}
\affiliation{Department of Physics, University of Queensland, St Lucia, Queensland 4072, Australia.}

\author{Nicolas C. Menicucci}
\affiliation{Department of Physics, University of Queensland, St Lucia, Queensland 4072, Australia.}
\affiliation{Department of Physics, Princeton University, Princeton, New Jersey 08544, USA}
\affiliation{Perimeter Institute for Theoretical Physics, Waterloo, Ontario N2L 2Y5, Canada}

\author{Timothy C. Ralph}
\affiliation{Department of Physics, University of Queensland, St Lucia, Queensland 4072, Australia.}

\author{Peter van Loock}
\affiliation{Optical Quantum Information Theory Group, Max Planck Institute for the Science of Light,
Institute of Theoretical Physics I, Universit\"{a}t Erlangen-N\"{u}rnberg,
Staudtstr.7/B2, 91058 Erlangen, Germany.}

\date{\today}

\begin{abstract}
Continuous-variable cluster states offer a potentially promising method of implementing a quantum computer. This paper extends and further refines theoretical foundations and protocols for experimental implementation. We give a cluster-state implementation of the cubic phase gate through photon detection, which, together with homodyne detection, facilitates universal quantum computation. In addition, we characterize the offline squeezed resources required to generate an arbitrary graph state through passive linear optics. Most significantly, we prove that there are
universal states for which the offline squeezing per mode does not increase with the size of the cluster.
Simple representations of continuous-variable graph states are introduced to analyze graph state transformations under measurement and the existence of universal continuous-variable resource states.
\end{abstract}


\maketitle

\section{Introduction}

Nonstandard models of quantum computation are important, both practically and conceptually. On the one hand, they lead to new experimental methods to realize quantum computers;
on the other hand, they offer additional insight on the often counterintuitive properties of quantum information. Continuous-variable (CV) quantum computation~\cite{Lloyd1999} not only provides a framework for description of interacting quantum fields~\cite{Braunstein2005a}, but also offers additional realizations of quantum computers when each CV mode is assigned a suitable qubit encoding~\cite{Gottesman2001,Ralph2003}. Meanwhile, cluster-state computation~\cite{Raussendorf2001} showed that the implementation of many difficult Hamiltonians may be avoided by just applying single-qubit measurements on a suitably prepared multi-party entangled resource state;
hence challenging traditional intuition that the implementation of a general unitary operator requires unitary evolution.

CV cluster-state computation is a fusion of these protocols~\cite{Zhang2006,Menicucci2006}. In addition to its intrinsic conceptual interest, the formalism presents a potential alternative implementation of a quantum computer. Optical CV cluster states have distinct advantages over discrete analogues~\cite{Nielsen2004}.  Any such cluster state may be generated deterministically through offline squeezing and passive linear optics~\cite{vanLoock2007}, while all multi-mode Gaussian transformations performed through the cluster
require only homodyne detection~\cite{Menicucci2006}.  In addition,
via alternative techniques, large CV clusters can be generated in a single step using just one optical parametric oscillator (OPO) and no interferometer~\cite{Menicucci2007}; some such proposals also have significant scaling potential~\cite{Menicucci2008,Flammia2008a}.  These features of CV cluster states suggest that they offer a fertile experimental testing ground for the principles of measurement-based computation~\cite{Raussendorf2003}.  CV cluster states involving four optical modes have been demonstrated experimentally~\cite{Yonezawa2008,Yukawa2008,Su2007}.

In this paper, we expand and extend the results given in Ref.~\cite{Menicucci2006}. First, we apply the CV stabilizer formalism~\cite{vanLoock2007,Barnes2004} to give simple phase-space and algebraic representations of CV graph states. We then apply these tools to compute how graph states transform through quadrature measurements and show that there exist universal graph states---cluster states---that can be used as resource states for the implementation of an arbitrary CV circuit~\footnote{In this article, a `graph state' can have an arbitrary graph, while a `cluster state' must be a member of a family of graph states that is universal for quantum computation.  The reader should be aware that conventions vary in the literature, and these terms are sometimes used interchangeably.  We will, on occasion, use the term `cluster' on its own, whose meaning at the time should be clear from the context.}.

Second, we extend the results of Ref.~\cite{vanLoock2007} by bounding the offline squeezed resources required to construct an arbitrary graph state to a given precision through passive linear optics. These results are applied to several graph states of common interest, including linear graph states and universal cluster states. We show that
{\it the level of squeezing required per mode does not grow with the size of the cluster state};
a necessary criterion to perform quantum computation efficiently through offline resources. In addition, we prove that even if online squeezing is assumed to be as readily available as its offline counterpart, the generation of CV cluster states via offline resources remains less costly.

Third, we detail an explicit optical implementation of a non-Gaussian operator through photon counting and homodyne measurements and thus propose an explicit measurement sequence on CV cluster states that facilitates universal quantum computation. We also present an alternative formalism such that the embedding of non-Gaussian resource states allows for universal quantum computation entirely by homodyne measurements alone. Together these results refine many of the details of the CV formalism, offer tools for further development of CV cluster-state protocols, and present a variety of potentially promising and viable experiments.

The structure of the paper is as follows. Section II describes background material on CV quantum computation and qubit cluster states that will be useful later in the paper. Section III introduces graph states for CV modes (qumodes) and describes their stabilizer and phase-space representations. Section IV demonstrates that such states, coupled with single-qumode measurements, are capable of implementing any specific unitary. Section V explores how CV graph states transform under measurements and applies these results to construct a CV cluster state that may be used as a resource for universal quantum computation. In Section VI, the case of imperfect CV clusters and the resulting distortions are analyzed and discussed. Section VII discusses the optical implementation of CV cluster-state computing, including the resource requirements for generating arbitrary CV clusters and explicit implementation of a nonlinear gate that facilitates universal quantum computation. Section VIII concludes the paper.

\section{Preliminaries}
In this section, we review some of the background knowledge relevant to CV cluster-state computation, and its optical implementation. Familiarity with quantum computation and quantum optics to the level of Refs.~\cite{Nielsen2000,Gerry2005} is assumed. CV cluster-state computation combines the concepts of CV quantum computation and cluster states. For a more extensive review of these topics, please see Refs.~\cite{Braunstein2005a,Raussendorf2003}.

\subsection{Continuous-Variable Quantum Computation}

\subsubsection{CV State Representations}

In traditional quantum computation, which uses discrete quantum variables, the basic unit of information is the qubit, a system with a two-dimensional Hilbert space with computational basis states $\ket{0}$ and $\ket{1}$ and conjugate basis states~$\ket{+}$ and~$\ket{-}$.  The two bases are related by the Hadamard operation~$H$.

The analogue for CV quantum computation is the qumode~\footnote{We use the terms `mode' and `qumode' interchangeably.}, a quantum system with an infinite-dimensional Hilbert space spanned by a continuum of orthogonal states $\ket{s}_q$ for each $s \in \mathbb{R}$, with orthogonality condition $\bra{r}_q\ket{s}_q = \delta(r-s)$.  The conjugate basis states are labeled~$\ket{s}_p$.  The two bases are related by a Fourier transform operation:
\begin{align}\nonumber
\ket{s}_p &= \frac{1}{\sqrt{2 \pi}}\int_{-\infty}^\infty dr\, e^{i r s} \ket{r}_q = F \ket{s}_q, \\ \label{eq:dftstates}
\ket{s}_q  &= \frac{1}{\sqrt{2 \pi}}\int_{-\infty}^\infty dr\, e^{- i r s} \ket{r}_p = F^{\dag} \ket{s}_p.
\end{align}
The unitary operator~$F$ is defined by this relation.    %
%
%
In quantum protocols, qumodes may be used to encode qubits (e.g.,\ the GKP encoding~\cite{Gottesman2001} or a coherent-state encoding~\cite{Ralph2003}), or they may be employed directly for CV quantum computation~\cite{Lloyd1999,Bartlett2002}.

We may now define corresponding observables, position~$\hat{q}$ and momentum~$\hat{p}$, such that $\hat{q}\ket{s}_q = s \ket{s}_q$ and $\hat{p}\ket{s}_p = s \ket{s}_p$, with $[\hat{q},\hat{p}] = i$ where $\hbar = 1$.  Here, $\hat{p}$ is the generator of positive translations in position, while $-\hat{q}$ is the generator of positive translations in momentum. Thus, we can write an arbitrary position and momentum eigenstate as
\begin{align}
\ket{s}_q = X(s) \ket{0}_q, \qquad \ket{s}_p = Z(s) \ket{0}_p,
\end{align}
where $X(s) = e^{-is\hat{p}}$ and $Z(s) = e^{is\hat{q}}$ represent displacements in the computational and conjugate basis, respectively.  An arbitrary pure quantum state $\ket{\phi}$ of a CV system may be decomposed as a superposition of either $\ket{s}_p$ or $\ket{s}_q$.

While the computational basis or its conjugate is uncountable, any physical state $\ket{\phi}$ may nevertheless be decomposed into a countably infinite basis. For particles in a harmonic trap or quantum optical fields we can use the Fock basis of definite particle number $\{\ket{0}, \ket{1}, \ldots\}$ where $\hat{n} = \hat{a}^\dag \hat{a}$ is the number operator, with $\hat{n} \ket{n} = n \ket{n}$, the usual bosonic commutator $[\hat{a},\hat{a}^\dag]=1$, and $\hat{a} = (\hat{q} + i \hat{p})/\sqrt{2}$.  In the terminology of quantum optics, $\hat{q}$ and $\hat{p}$ are referred to as the `position quadrature' and `momentum quadrature' for a given mode, respectively.

A qumode is in a minimum uncertainty state if the product of the quadrature deviations $\Delta \hat{q}$ and $\Delta \hat{p}$ is minimized, i.e.,\  $\Delta \hat{q} \Delta \hat{p} = \frac{1}{2}$. The ground or vacuum state $\ket{0}$ defined by $\hat{a} \ket{0} = 0$ is an example of particular theoretical and practical interest and represents a Gaussian superposition centered about $0$ in either the computational or the conjugate basis:
\begin{equation}
\ket{0} = \frac{1}{\pi^{1/4}}\int ds\, e^{-s^2/2} \ket{s}_q = \frac{1}{\pi^{1/4}}\int ds\, e^{-s^2/2} \ket{s}_p.
\end{equation}
The vacuum state is a specific example of a {\it Gaussian state} whose quadratures exhibit Gaussian statistics.

The state of
a single qumode can be described by its Wigner Function~\cite{Leonhardt1997}
\begin{equation}\label{eqn:wignerdef}
W(x,y) = \frac{1}{2 \pi} \int dw \bra{x - \frac{w}{2}}_q\hat{\rho}\ket{x+ \frac{w}{2}}_q e^{iwy}.
\end{equation}
The Wigner function is a useful tool for describing arbitrary Gaussian states, which
are completely determined by the first and second moments of the quadratures.
Any state with a Gaussian Wigner function is, by definition, a Gaussian state.  For instance, the Wigner function of the vacuum state is $e^{-(x^2 + y^2)}/\pi$, a multivariate Gaussian distribution with a variance of~$1/2$ in both quadratures. A multi-mode state such as a CV cluster state is described by a multi-mode Wigner function,
a straightforward extension of Eq.~(\ref{eqn:wignerdef}). Multi-mode Gaussian states are then given
by a second-moment covariance matrix and a first-moment vector \cite{Braunstein2005a}.

\subsubsection{Gaussian Transformations}

In quantum optics, the Hamiltonians corresponding to the experimentally most feasible interactions are at most quadratic in $\hat{q}$ and $\hat{p}$. Such interactions transform Gaussian states to Gaussian states, and are referred to as Gaussian transformations or linear unitary Bogoliubov (LUBO) transformations.  If we collect the quadrature operators into an operator-valued vector $\hat{\mathbf{v}} = (\hat{q}_1,\hat{q}_2,\ldots,\hat{p}_1,\hat{p}_2,\ldots)^T$, then a general Gaussian transformation~$\hat{U}$ transforms $\hat{\mathbf{v}}$ according to:
\begin{equation}\label{eqn:gauss}
\hat U^{\dag} \hat{\mathbf{v}} \hat U  =  L \hat{\mathbf{v}} + \mathbf{c}, \qquad \Det (L) = 1,
\end{equation}
where $L$ is a $2n \times 2n$ symplectic matrix and $\mathbf{c}$ is a vector of $2n$ constants that represent quadrature displacements. We list a number of standard {\it single-mode} Gaussian transformations that will be used in this paper, along with their associated Heisenberg action on the quadrature operators.

\begin{description}
\item[(a) Rotations:] $R(\theta) = e^{i\theta(\hat{q}^2 + \hat{p}^2)/2}$ rotates a state counterclockwise in phase space by an angle $\theta$. Rotations are also referred to as phase shifts.
\begin{align}
 \left(\begin{array}{c}
\hat{q}    \\
\hat{p}     \\
\end{array} \right)  &\rightarrow   \left(\begin{array}{cc}
\cos\theta  & -\sin\theta \\
\sin\theta  & \cos\theta \\
\end{array} \right) \left(\begin{array}{c}
\hat{q}    \\
\hat{p}     \\
\end{array} \right) = M_R(\theta)  \left(\begin{array}{c}
\hat{q}    \\
\hat{p}     \\
\end{array} \right) \label{eqn:rotations}
\end{align}
where $M_R(\theta)$ is the rotation matrix describing the linear Heisenberg action on the quadrature operators. Note that $R(\pi/2) = F$.

\item[(b) Quadrature Displacements:] $Z(s) = e^{is \hat{q}}$ displaces a state in phase space by $s$ in momentum.
\begin{align}
 \left(\begin{array}{c}
\hat{q}    \\
\hat{p}     \\
\end{array} \right)  \rightarrow  \left(\begin{array}{c}
\hat{q}    \\
\hat{p}     \\
\end{array} \right) + \left(\begin{array}{c}
0    \\
s     \\
\end{array} \right)
\end{align}
Similarly, $X(s) = e^{-is\hat p}$ displaces a state in phase space by $s$ in position.  Note the sign in the exponential of each.

\item[(c) Squeezing:] $S(s) = e^{-i\ln(s)(\hat{q}\hat{p}+\hat{p}\hat{q})/2}$ squeezes the position quadrature by a factor of $s$, while stretching the conjugate quadrature by $1/s$.
\begin{align}
\left(\begin{array}{c}\nonumber
\hat{q}    \\
\hat{p}     \\
\end{array} \right)  &\rightarrow   \left(\begin{array}{cc}
s  & 0\\
0  & 1/s \\
\end{array} \right) \left(\begin{array}{c}
\hat{q}    \\
\hat{p}     \\
\end{array} \right) = M_S(s)  \left(\begin{array}{c}
\hat{q}    \\
\hat{p}     \\
\end{array} \right)
\end{align}
where $M_S(s)$ is the squeeze matrix describing the linear Heisenberg action on the quadrature operators.

\item[(d) Shearing:] $D_{2,\hat q}(s) = e^{is\hat{q}^2/2}$ shears a state with respect to the $\hat{q}$ axis by a gradient of $s$. The shearing operator $e^{is\hat{q}^2/2}$ is also referred to as the phase gate.
\begin{align}\nonumber
\left(\begin{array}{c}
\hat{q}    \\
\hat{p}     \\
\end{array} \right)  &\rightarrow   \left(\begin{array}{cc}
1  & 0\\
s  & 1 \\
\end{array} \right) \left(\begin{array}{c}
\hat{q}    \\
\hat{p}     \\
\end{array} \right) = M_D(s)  \left(\begin{array}{c}
\hat{q}    \\
\hat{p}     \\
\end{array} \right)
\end{align}
where $M_D(s)$ is the shearing matrix describing the linear Heisenberg action on the quadrature operators.

\end{description}
%

Operations~(a) and~(b) correspond to the most readily available single-mode Gaussian transformations, requiring only phase shifts and coherent-state sources.  To access all possible single-mode Gaussian transformations, we will need squeezing interactions to stretch and compress phase space uncertainties. Two such operations are given by~(c) and~(d). In experimental quantum optics, such interactions require nonlinear optical processes
(while the Heisenberg in-out relations remain linear). Typical methods involve optical parametric amplification, which allows one to generate squeezed vacuum states~$S(s)\ket{0}$. We refer to such processes as offline squeezing,
solely involving the preparation of squeezed vacuum states.

Offline squeezing contrasts with the online squeezing, where the squeezing operator is applied ``on-line'' to an arbitrary state of the electromagnetic field.  In experimental quantum optics, it is common to refer to $S(s) \ket{0}$ as a state with $10 \log (s^2)~\text{dB}$ of squeezing, alluding to the view that squeezing can be regarded as a physical resource~\cite{Braunstein2005a}.  While the generation of reasonably high levels ($10~\text{dB}$) of offline squeezing can be experimentally achieved~\cite{Vahlbruch2008}, online squeezing~\cite{LaPorta1989} is far more demanding,
and is currently only experimentally viable for modest values of $s$, for instance, online squeezing of $2.5~\text{dB}$~\cite{Yoshikawa2007}
utilizing offline squeezed ancilla states~\cite{Filip2005}.

An arbitrary single-mode Gaussian transformation may be decomposed into (a) rotations, (b) quadrature displacements, and either (c) squeezing or (d) shearing operations. The addition of a two-mode Gaussian gate, such as a beamsplitter or $\CZ = e^{i\hat{q} \otimes \hat{q}}$, allows for arbitrary {\it multi-mode} Gaussian transformations.
To account for imperfect Gaussian transformations, e.g. affected by photon losses and thermal excess noises,
Gaussian unitary gates are generalized to Gaussian operations (Gaussian completely positive maps)~\cite{EisertPlenio2003}.
These also include Gaussian measurements such as homodyne detection.
Any quantum evolution consisting solely of Gaussian operations on Gaussian states may be efficiently simulated on a classical computer~\cite{Bartlett2002}.  Therefore, some sort of non-Gaussian element is required for universal quantum computation.  In fact, at least in principle, \emph{any} such element will do~\cite{Lloyd1999}.

\subsubsection{A Universal Gate Set}

We follow the definition of universal CV quantum computation outlined in Ref.~\cite{Lloyd1999}. A system is universal if it can simulate the action of a Hamiltonian consisting of a general polynomial of $\hat{p}$ and $\hat{q}$ to any fixed accuracy.

For a single qumode, all Gaussian operations together with any single nonlinear (non-Gaussian, at least cubic) interaction are capable of universality~\cite{Lloyd1999}. For example, the set of gates $D_{k,\hat{q}}(s) = \exp (is\hat{q}^k/k)$, for $k = 1,2,3$ for all~$s \in \reals$, together with the Fourier transform~$F$, are sufficient for universal single-mode quantum computation (that is, this set can be used to implement any single-mode unitary operation up to arbitrary accuracy).  Here $D_{1,\hat{q}}(s)$ is a displacement, $D_{2,\hat{q}}(s)$ is a shear, and $D_{3,\hat{q}}(s)$ is the cubic phase gate~\cite{Gottesman2001}. Adding to this set any nontrivial two-mode interaction allows for universal quantum computation.  For theoretical simplicity, here we shall use the $\CZ$~gate (defined above) to complete the universal set, while another possibility is a simple beamsplitter interaction.

It should be noted that such statements about universality do not account for noise. Presently, all general error correction codes require discretization at some level. Hence, currently CV quantum computation is only known to be possible for discretized encodings of CVs.

\subsection{Cluster-State Computation}

In the qubit-based cluster-state model~\cite{Raussendorf2001}, quantum computation is implemented by a series of single-qubit measurements on a specially prepared, entangled state of many qubits, most generally, referred to as a graph state~\cite{Hein2004}. Such states may be conveniently described by graphs. A graph $G = (V,E)$ consists of a vertex set $V = \{v_i\}_{i=1}^n$ and a set of edges $E$. We say that two vertices, $v_i$ and~$v_j$, are neighbors if there exists an edge $(v_i,v_j) \in E$ that connects them. For an introduction to graphs, see Ref.~\cite{West2000}.

For any undirected, unweighted graph $G = (V,E)$ having no self-loops, we can construct a corresponding graph state as follows.   For each vertex of $G$, we initialize a qubit in the state $\ket{+} = \frac{1}{\sqrt2}(\ket{0} + \ket{1})$. For every edge in $G$ linking two vertices, we apply a CSIGN gate (which is sometimes called the CPHASE gate) to the two corresponding qubits. Any unitary operation can be implemented on a tailor-made graph state using an appropriate sequence of single-qubit measurements.

The stabilizer formalism~\cite{Gottesman1998a} offers an efficient way to represent any graph state.  A state $\ket{\phi}$ is stabilized by an operator $K$ if it is an eigenstate of $K$ with unit eigenvalue, i.e.,\ $K \ket{\phi} = \ket{\phi}$. The set of stabilizers form an abelian group under operator multiplication.  If such a set exists for a given state, then we call that state a stabilizer state, and we may use the generators of its stabilizer group to uniquely specify it. The stabilizers for qubit graph states are well known. Given that $\ket{\phi}$ is an $n$-qubit graph state with associated graph $G = (V,E)$, it is stabilized by
\begin{equation}\label{eqn:discrete_stab}
K_i =  X_i \prod_{j \in N(i)} Z_j
\end{equation}
where $N(i)$ denotes the set of indices that define the set of vertices that neighbor $v_i$, i.e., $N(i) = \{j \mid (v_j,v_i) \in E$\}.
The operators $X$ and $Z$ are the usual Pauli operators for qubits.

There exist universal families of graph states that may be used to implement any unitary operation solely through the choice of single-qubit measurements made on it. Originally, such states are called cluster states, and cluster-state computation involves the implementation of arbitrary algorithms solely by an adaptive local measurement scheme.  The scheme involves only single-qubit measurements and is called `adaptive', because the choice of measurement bases depends both on the algorithm to be implemented and, in general, on the measurement outcomes during a
cluster computation.
Cluster states, when combined with adaptive local measurements, are thus universal resources for quantum computation~\cite{Raussendorf2001}.
For more recent developments on possible resource states for universal quantum computation and their requirements, see Refs.~\cite{GrossPRL2007,GrossPRA2007,vandenNest2006,vandenNest2007}.

\section{Continuous-Variable Graph States}

The concepts of qubit cluster-state computation can be extended to the continuous-variable regime. We outline CV graph states~\cite{Zhang2006}, which can be used as resource states for universal CV quantum computation~\cite{Menicucci2006}.  We then introduce nullifiers, a variation of the CV stabilizer formalism~\cite{vanLoock2007,Barnes2004}, and use them to compute how CV graph states transform under quadrature measurements.

The basic premise of CV graph states may be obtained by replacing elements of qubit cluster-state computation with their CV analogues: $\ket{+}$ becomes $\ket{0}_p$, $X$-measurements are replaced by measurements of $\hat{p}$ (and $Z$ with $\hat{q}$), and the CSIGN interaction is replaced by the $\CZ = e^{i\hat{q}_i \hat{q}_j}$ gate, which is used to entangle nodes~$i$ and~$j$. Each CV graph state can also be defined by a graph $G = (V,E)$, where the set of vertices $V$ corresponds to the individual qumodes, and the edge set $E$ determines which qumodes interact via the $\CZ$ operation.

It should be mentioned that one way to generalize the idea of a CV graph state is to use weighted edges for the graph.  The edge weights specify the strength of the $\CZ$ interaction between the connected nodes: $\CZ(t) = e^{it \hat{q} \otimes \hat{q}}$, where $t$~is the edge weight.  CV weighted graph states have a variety of uses~\cite{Zhang2009,Zhang2008a,Menicucci2007}, but in this article we will confine further discussion to unweighted graphs (or, equivalently, graphs whose edge weights are all $+1$).

\subsection{Stabilizers and Nullifiers}

Analogous to the case for qubit graph states, the stabilizer formalism for CV systems~\cite{vanLoock2007,Barnes2004} can be used to specify any CV graph state completely~\cite{Zhang2008a}.  We say that a zero-momentum eigenstate $\ket{0}_p$ is stabilized by $X(s)$ for all~$s$, since it is a $+1$-eigenstate of those operators.  This holds even though $X(s)$, being non-Hermitian, is not an observable.  Notice that if $K$ stabilizes $\ket{\phi}$, then $U K U^\dag$ stabilizes $U \ket{\phi}$. This observation, together with the relation $e^{i\hat{q}_1\hat{q}_2} \hat{p}_1 e^{-i\hat{q}_1\hat{q}_2} = \hat{p}_1 - \hat{q}_2$, allow us to write the stabilizers for an arbitrary CV graph state $\ket{\phi}$ on $n$ qumodes with graph $G = (V,E)$:
\begin{equation}\label{eqn:cont_stab}
K_i(s) =  X_i(s) \prod_{j \in N(i)} Z_j(s), \qquad i = 1,\ldots n
\end{equation}
for all $s \in \mathbb{R}$, where $N(i)$ is defined as before in Eq.~\eqref{eqn:discrete_stab}, and the subscript indicates which qumode the displacement acts on.

This group is conveniently defined by its Lie algebra, the space of operators $H$ such that $H \ket{\phi} = 0$. We refer to any element of this algebra as a \emph{nullifier} of $\ket{\phi}$ and the entire algebra as the \emph{nullifier space} of~$\ket{\phi}$.  Being Hermitian, every nullifier is an observable.  Any ideal graph state has a simple nullifier representation.

\begin{thrm}
The nullifier space of an $n$-qumode graph state $\ket{\phi}$ with graph $G = (V,E)$ is an $n$-dimensional vector space spanned by the following Hermitian operators:
\begin{equation}\label{eqn:nullifers}
H_i = \hat{p}_i - \sum_{j \in N(i) } \hat{q}_j \qquad i = 1,\ldots,n.
\end{equation}
That is, any linear superposition $H = \sum_i c_i H_i$ satisfies $H \ket{\phi} = 0$. Note that $[H_i,H_j] = 0$ for all $(i,j)$.
\end{thrm}

\begin{proof}
Every stabilizer from Eq.~\eqref{eqn:cont_stab} is the exponential of a nullifier in this space.  Specifically, $K_i(s) = e^{-isH_i}$ for all~$s \in \reals$, with $i = 1, \dotsc, n$.  $\blacksquare$
\end{proof}

\begin{figure}[tb]
\center
\includegraphics[width=\columnwidth]{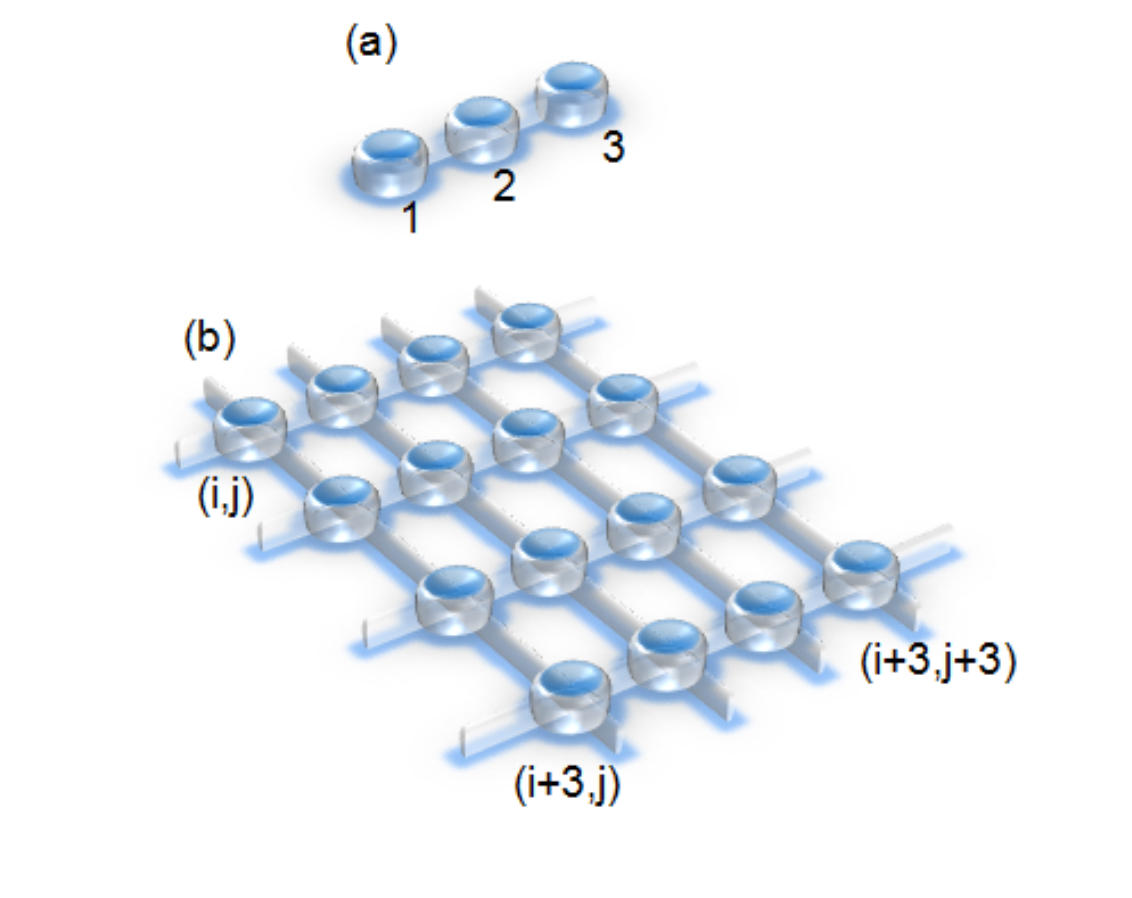}
\caption{\label{fig:graph_states} Nullifiers give an efficient description of ideal graph states. The nullifier space of the linear graph state on three nodes (a) is spanned by $\hat{p}_1 - \hat{q}_2$, $\hat{p}_2 - \hat{q}_1 - \hat{q}_3$ and $\hat{p}_3 - \hat{q}_2$. The infinite two-dimensional lattice (b) is nullified by $H_{i,j} = \hat{p}_{i,j} - \hat{q}_{i-1,j} - \hat{q}_{i+1,j} - \hat{q}_{i,j-1} - \hat{q}_{i,j+1}$ for each coordinate $(i,j)$.}
\end{figure}

In Fig.~\ref{fig:graph_states}, we illustrate this formalism. Note that the nullifier space for a given state does not have a unique set of nullifiers (a basis) that defines it, since linear combinations of nullifiers give another nullifier.  Nevertheless, Eq.~\eqref{eqn:nullifers} is a standard set that can easily be obtained from a given graph.

\subsection{Wigner Representation}

The Wigner function can be useful as an extension to the nullifier formalism. It encapsulates the simplicity of the nullifier formalism, while maintaining the intuition afforded by an explicit representation of the state, and importantly, continues to be useful for non-ideal CV cluster states. Since the arguments of a Wigner function behave identically to the nullifiers under Gaussian transformations, they may also be written by inspection.

The Wigner function of an ideal $n$-qumode graph state $\ket{\phi}$ with graph $G = (V,E)$ is a function of $2n$ variables on the scalar-valued vectors $\mathbf{q} = (q_1,\ldots,q_n)$ and $\mathbf{p} = (p_1,\ldots,p_n)$. Explicitly,
\begin{equation}\label{eqn:wigner}
W(\mathbf{q},\mathbf{p}) = \prod_{i = 1}^n \epsilon(q_i) \delta (H_i)
\end{equation}
where $H_i$, $i = 1,\ldots,n$ are the standard nullifiers of $\ket{\phi}$ (with each of the operators $\hat{q}_i$ and $\hat{p}_i$ replaced by scalar variables $q_i$ and $p_i$, respectively), $\delta(x)$ is the Dirac delta-distribution, and $\epsilon(x)$~is the infinite uniform distribution.  Ideal CV graph states are highly singular, so for all practical purposes, $\delta(x)$ and $\epsilon(x)$ should be considered limits of a normalized Gaussian whose variance, respectively, vanishes and extends to infinity.  For example, the Wigner function of Fig.~\ref{fig:graph_states}(a) is $\epsilon(q_1) \epsilon(q_2) \epsilon(q_3) \delta(p_1 - q_2)\delta(p_2 - q_1 - q_3)\delta(p_3 - q_2)$.

Wigner functions can also be used to define an extended class of generalized graph states. Whereas an ideal graph state with associated graph $G$ may be defined by the action of appropriate $\CZ$ gates on $n$ momentum eigenstates, a generalized graph state replaces each momentum eigenstate with some arbitrary quantum state $\ket{\phi_i}$. If $\ket{\phi_i}$ has a corresponding Wigner function $W_i(q_i,p_i)$, then the Wigner representation of the resulting generalized graph state is given by
\begin{equation}\label{eqn:general_cluster}
W(\mathbf{q},\mathbf{p}) = \prod_{i = 1}^n W_i(q_i,H_i).
\end{equation}
Such states are used extensively when we perform computations with graph states and when we extend the graph state formalism to realistic situations where momentum eigenstates need to be approximated.

\section{Quantum Computation on CV Graph States}

CV graph states are a resource for CV quantum computation. For any given CV unitary $U$, and any given input $\ket{\phi}$, there exists an appropriate graph state such that by entangling the graph state locally with $\ket{\phi}$ and applying an appropriate sequence of single-qumode measurements, $U \ket{\phi}$ is computed.

To justify this statement, we first show that there exists a $\ket{\phi}$-dependent quantum state on a system of qumodes that collapses into $U \ket{\phi}$ (modulo known single qumode Gaussian operations) when an appropriate sequence of single-qumode measurements is applied. We then demonstrate that this $\ket{\phi}$-dependent quantum state can be efficiently constructed using an appropriate graph state as a resource.

\subsection{Measurement-Based CV Quantum Computation}

To implement any unitary operation on $k$ qumodes, we apply the following algorithm. We first introduce a graph $G = (V,E)$. We designate $k$ vertices of $G$ as input vertices, and another (possibly overlapping) set as output. Call these sets $V_{\text{in}}$ and $V_{\text{out}}$. The following algorithm computes $U \ket{\phi}$:
\begin{enumerate}
\item The qumodes corresponding to the vertices in~$V_{\text{in}}$ encode the input state $\ket{\phi}$, while the qumodes corresponding to the other vertices are each initialized in the state $\ket{0}_p$.
\item For each edge $(v_j,v_k) \in E$, apply $\CZ = e^{i\hat{q}_j\hat{q}_k}$ between vertices $j$ and $k$. Since all $\CZ$ operations commute, their order does not matter.
\item Measure each vertex $v_i$, for all $v_i \not\in V_{\text{out}}$ in a basis of the form~$M_i = e^{-if_i(\hat{q})}\hat{p}e^{if_i(\hat{q})}$, where $f_i(\hat{q})$ is, in general, a polynomial of $\hat{q}$. The exact form of each~$f_i$ is dictated by the unitary we wish to implement and the result of measurements on prior modes. Without loss of generality, we can label the vertices such that they are measured in numerical order.
\item The remaining unmeasured qumodes encode $U \ket{\phi}$, modulo known single-mode rotations and translations.
\end{enumerate}
The above algorithm may be implemented by using an appropriate graph state as a resource. This algorithm is universal. Given any unitary $U$, there always exists an appropriate graph $G = (V,E)$ and designations $V_{\text{in}}, V_{\text{out}} \subseteq V$ such that the above algorithm implements $U$.

\subsection{Proof of Universality}\label{sec:m_comp}

To prove the above procedure is universal, we need to show that it can implement (a) single-mode Gaussian operations, (b) the cubic phase gate, and  (c) the $\CZ$ gate. First observe that (c) may be implemented trivially by a two vertex graph where both vertices are designated as both input and output. No measurements are required.

The implementation of (a) and (b) also each involve a two-vertex graph. We designate one vertex as input and the other as output. Consider first the case where the input mode is measured in the $\hat{p}$ basis:
\[
\Qcircuit @C=1em @R=1em {
    \lstick{\ket{\phi}} & \ctrl{1} &  \qw & \measureD{\hat{p}} & \rstick{m} \cw
     \\
    \lstick{\ket{0}_p} & \ctrl{0} & \qw & \rstick{X(m) F \ket{\phi}} \qw
}
\]
Given an input $\ket{\phi} \otimes \ket{0}_p = \int ds\, f(s) \ket{s}_q \ket{0}_p$, the state of the system after the application of the $\CZ$ gate is
\begin{align}
	\CZ \bigl(\ket{\phi} \otimes \ket{0}_p\bigr) = \int ds\, f(s) \ket{s}_q \ket{s}_p.
\end{align}
Measurement of $\hat{p}$ on the first mode with associated result $m$, as shown, collapses this state to
\begin{align} \nonumber
\ket{\phi}_\mathrm{out} &\propto \int ds\, f(s) \Bigl(\bra{m}_p \ket{s}_q \Bigr) \ket{s}_p \\
&\propto \int ds\, f(s) e^{-ism} \ket{s}_p, \nonumber \\
\ket{\phi}_\mathrm{out} & = e^{-im \hat{p}} \int ds\, f(s) \ket{s}_p = X(m) F \ket{\phi}.
\end{align}
The effect of this circuit is to apply the identity operation, modulo a known quadrature rotation and displacement. Obviously, a transformed input state of $D_{\hat{q}} \ket{\phi}$, for any $D_{\hat{q}} = e^{if(\hat{q})}$ diagonal in the computational basis, would result in output $X(m) F D_{\hat{q}} \ket{\phi}$. However, this same transformation can be induced by an appropriate measurement. We can see this immediately by writing out the associated circuit.  Since $D_{\hat{q}}$ commutes with $\CZ$, the circuit
\[
\Qcircuit @C=1em @R=1em {
    \lstick{\ket{\phi}} & \ctrl{1} &  \gate{D_{\hat{q}}} & \measureD{\hat{p}} & \rstick{m} \cw
     \\
    \lstick{\ket{0}_p} & \ctrl{0} & \qw & \rstick{X(m) F D_{\hat{q}} \ket{\phi}} \qw
}
\]
must have the desired output. The application of $D_{\hat{q}}$ followed by a $\hat{p}$-measurement has an identical effect to a measurement in the rotated basis $\hat{p}_{f(\hat{q})} = D^\dag_{\hat{q}} \hat{p} D_{\hat{q}}$. Hence, any unitary diagonal in the computational basis can be implemented by a single measurement of $\hat{p}_{f(\hat{q})}$. Measurements on two qumodes of a three-qumode cluster equates to a repeated application of this circuit, resulting in the output
\begin{align}
\nonumber
\ket{\phi}_{\text{out}} &= X(m_2) F D_{\hat{q}} X(m_1) F D_{\hat{q}} \ket{\phi}\\ \nonumber
&=   X(m_2) F X(m_1) D_{(\hat{q}+m_1)} F D_{\hat{q}} \ket{\phi}\\
&=   X(m_2) F X(m_1) F D_{(-\hat{p}+m_1)} D_{\hat{q}} \ket{\phi}, \label{eqn:single_circuit}
\end{align}
where $D_{\hat v} = e^{if(\hat{v})}$ for any operator~$\hat{v}$, and we have used the relations
\begin{align}
	X(-m) \hat{q} X(m) &= \hat{q} + m, \\
	Z(-m) \hat{p} Z(m) &= \hat{p} + m, \\
	F^\dag (-\hat{q}) F &= \hat{p}, \\
	F^\dag \hat{p} F &= \hat{q},
\end{align}
the last two of which give
\begin{align}
	F^\dag Z(m) F &= X(m), \\
	F^\dag X(m) F &= Z(-m).
\end{align}
If, instead of $\hat{p}_{f(\hat{q})}$ as shown above, we had measured the second mode in the outcome-dependent basis~$\hat{p}_{f(-\hat{q} - m_1)}$, the output would instead be $\ket{\phi}_{\text{out}} = X(m_2) F X(m_1) F D_{\hat{p}} D_{\hat{q}} \ket{\phi}$.  Thus, the ability to measure the second mode in the basis $\hat{p}_{f(-\hat{q} - m_1)}$ allows deterministic implementation of $D_{\hat{p}}$.

By concatenating this circuit a sufficient number of times, any single-mode operation can be implemented deterministically by alternating application of $D_{\hat{q}}$ and $D_{\hat{p}}$ (with generally different $D$s each time)~\cite{Lloyd1999}.  Note that the measurements required to implement these gates (beyond the first one) are necessarily adaptive---that is, our choice of the measurement basis is generally dependent on previous measurement results. Also notice that the final result is modified by a measurement-dependent Gaussian operation ($X(m_2) F X(m_1) F$, in the simple case illustrated).  This need not be corrected.  As long as we keep track of it, it can instead be considered as just a change of basis for the final state.

Another useful way of approaching the question of universality in the CV context is to consider implementing Gaussian operations and then, separately, non-Gaussian operations.  Single-mode Gaussian operations require the ability to implement $e^{is\hat{q}}$ (quadrature displacements) and $e^{is\hat{q}/2}$ (shears) for all $s \in \reals$, plus the Fourier transform~$F$.

The Fourier transform is obtained for free simply through the Gaussian correction applied with each measurement.  Displacements~$e^{is\hat{q}}$ are easily implemented, as well: just measure $\hat{p}_{s \hat{q}} = \hat{p} + s$, which is the same as measuring~$\hat{p}$ and adding~$s$ to the result~\footnote{There is a sign error in the corresponding derivation in Ref.~\cite{Menicucci2006}.}.  In this case, dependence on previous measurement outcomes is trivial since~$\hat{p}_{s (\hat{q} + m)}$ is also equal to~$\hat{p} + s$, and thus, no adaptation is required at all.

Shearing transformations~$e^{is\hat{q}/2}$ correspond to measurements in the basis~$\hat{p}_{s\hat{q}^2/2} = \hat{p} + s\hat{q}$.  In the case that adaptation for previous measurements is required, the new measurement basis would be of the form~$\hat{p}_{s(\hat{q} + m)^2/2} = \hat{p} + s\hat{q} + ms$, which differs from the original basis only by a measurement-dependent constant.  This can be accounted for trivially by measuring in the original basis and adding~$ms$ to the result.

Since the ``adaptation'' required for previous measurement outcomes is trivial for all measurements necessary to implement Gaussian operations, such measurements may be made in any order---or simultaneously.  This property is known as \emph{parallelism}~\cite{Menicucci2006}.  Qubit cluster-state computation has an analogous property with the same name, whereby measurements that implement Clifford group operations can be performed in any order~\cite{Raussendorf2003}.

The above measurements allow for any Gaussian operation to be implemented.  But universality requires non-Gaussian operations as well~\cite{Bartlett2002}.  The cubic phase gate $e^{is\hat{q}^3/3}$ allows implementation of all single-mode non-Gaussian operations~\cite{Lloyd1999} and may be implemented via a measurement in the basis~$\hat{p}_{s\hat{q}^3/3} = \hat p + s\hat{q}^2$.  The difference with this gate is that when an adaptive implementation is required, the physical basis is now different: $\hat{p}_{s(\hat{q}+m)^3/3} = \hat{p} + s\hat{q}^2 + \tfrac 2 3 ms\hat{q} + \tfrac 1 3 m^2 s$, due to the presence of the noncommuting $m$-dependent term~$\tfrac 2 3 ms\hat{q}$.  Accounting for this difference requires physically changing the basis of measurement (unlike the last term~$\tfrac 1 3 m^2 s$, which can be eliminated simply by shifting the result).  As with qubit quantum computation, a general CV quantum computation will require adaptive measurements for the non-Gaussian (non-Clifford) part of the computation. What these measurements are in an experimental context will depend on the choice of the physical implementation. In section VII, we propose one possible method that uses photon counting.

A sequence of single-qumode unitaries and wires is implemented by a sequence of measurements on a linear graph state. $\CZ$ gates are implemented by edges between linear clusters (Fig.~\ref{fig:circuit_mapping}). Thus, we may apply the algorithm described to implement any given CV unitary on an arbitrary input state.  This proves universality.

\begin{figure}[tb]
\centering
\includegraphics[width=\columnwidth]{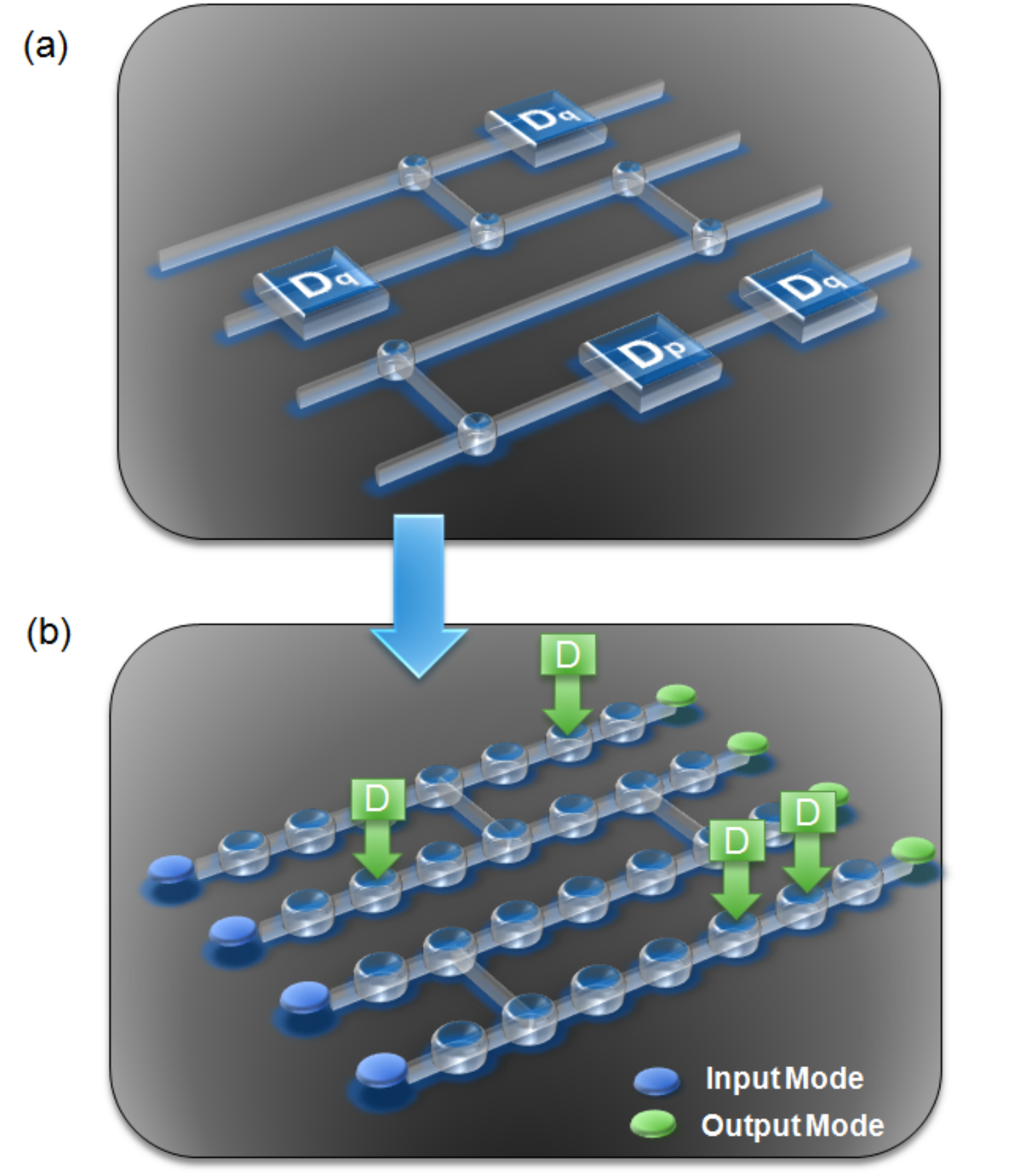}
\caption{\label{fig:circuit_mapping} (a)~Any unitary on multiple qumodes may be written as a quantum circuit consisting of $\CZ$ and single-qumode unitaries diagonal in either the position or momentum basis~\cite{Lloyd1999}. (b)~Any such circuit may be directly translated into an appropriate graph state. Here the arrowed qumodes are measured in the appropriate basis that implements their corresponding single-qumode unitary. All other non-output qumodes are measured in the $\hat{p}$ basis. }
\end{figure}

\subsection{Graph States as Resources}

Observe that steps~1 and~2 generate a special class of generalized CV graph states. Each of the input qumodes are initially set to encode the inputs of the desired quantum computation, rather than the standard momentum eigenstates. The resulting cluster has the Wigner representation
\begin{equation}\label{eqn:inputwigner}
W(\mathbf{q},\mathbf{p}) = \prod_{i \in \mathcal{I}} W_i(q_i,H_i) \prod_{j \not\in \mathcal{I}} \epsilon(q_j) \delta(H_j),
\end{equation}
where $\mathcal{I} = \{ i \mid v_i \in V_\mathrm{in} \}$ is the set of indices that corresponds to the input qumodes.

To see that the standard graph state with graph $G(V,E)$ may be used as a resource for the algorithm, we show that it may be used to efficiently generate clusters of the form  given by Eq.~(\ref{eqn:inputwigner}). Let the input state be initially encoded on $k$ qumodes, which we label $\{u_1,u_2,\ldots,u_k\}$, and the input vertices of the graph state given by $V_{\text{in}} = \{v_1,v_2,\ldots, v_k\}$. We proceed as follows:
\begin{enumerate}
\item Apply a $\CZ$ operation between each qumode pair, $(u_i,v_i)$, for each $i = 1,\ldots k$
\item Measure each of the modes $u_i$, resulting in measurement results $m_i$, $i = 1,\ldots,k$.
\end{enumerate}
The resulting cluster state is identical to Eq.~(\ref{eqn:inputwigner}), modulo known single-mode quadrature displacements and rotations on each $v_i$ that can be accounted for throughout the remainder of the computation.

Thus, the circuit of Fig.~\ref{fig:circuit_mapping}(a) may be implemented by using a standard graph state with the graph shown in Fig.~\ref{fig:circuit_mapping}(b).  We refer to such graph states as \emph{CV brickwork states}, alluding to similar results in qubit cluster-state computation~\cite{Broadbent2008}.

\section{Universal Cluster States}

So far, we have discussed the construction of specific graph states tailored for a specific quantum computation. Like qubit clusters, there exist classes of universal CV graph states that may be used to implement an arbitrary CV operation. This is of more than theoretical interest, since it facilitates the development of physical systems that are tailored to generate one  particular state. This state can then be used as a universal resource.

To prove the existence of such universal resources, we explore how CV graph states transform under single-mode quadrature measurements. These tools are then applied to show that there exists a universal CV graph state, which, when appropriate quadrature measurements are applied, collapses to the specific CV brickwork state that implements any given quantum circuit. Such universal graph states are called CV cluster states.

\subsection{Graph State Transformations}

The nullifier formalism is ideal for computing how graph states transform through quadrature measurements. In this formalism, we describe a measurement $\hat{p}_i$ on the $i^{th}$ qumode, with measurement result $m_i$, as follows. We first rewrite the nullifiers in a basis such that only one element, say $H_i$, does not commute with our basis of measurement. $H_i$ is then replaced with $\hat{p}_i - m_i$, and in all other nullifiers, $\hat{p}_i$ is replaced with $m_i$.  Measurements in the $\hat{q}_i$ basis are treated analogously. The details of this formalism are outlined in Appendix \ref{sec:nullifier}.

\subsubsection{Vertex Removal}

A computational-basis measurement on a qumode removes it, along with all edges that connect it, from the cluster. Consider a measurement $\hat{q}_i$ on the $i^{th}$ mode of $\ket{\phi}$. Equation (\ref{eqn:nullifers}) indicates that $H_i$ is the only noncommuting basis element. Therefore a measurement with result $m_i$ transforms $H_i$ into $\hat{q}_i - m_i$ and replaces $\hat{q}_i$ with~$m_i$ in all others.  Explicitly, for each~$j \neq i$,
\begin{align}
	H_j &\to H_j \Bigr\rvert_{\hat{q}_i \to m_i} \nonumber \\
	&=
	\begin{cases}
		\hat{p}_j - \sum_{(v_j,v_k) \in E, k\neq i} \hat{q}_k - m_i, & \text{if $(v_j,v_i) \in E$}, \\
		\hat{p}_j - \sum_{(v_j,v_k) \in E} \hat{q}_k, & \text{if $(v_j,v_i) \not\in E$}.
	\end{cases}
\end{align}
The resulting state corresponds to the graph state of $G$ with vertex $v_i$ removed, modulo known quadrature displacements. This operation is useful for creating a CV graph state that corresponds to the subgraph of some original resource state. In addition, it functions as a handy ``delete'' button, and can be used to ``amputate'' corrupted parts of a cluster state.
To summarize, if $\ket{\phi}$ is the graph state defined by a graph $G = (V,E)$, a $\hat{q}$-measurement on a mode $i$ results in the graph state with associated graph $G \backslash \{v_i\}$, modulo known corrections, i.e., the graph resulting from removal of vertex $v_i$ along with all edges connecting to $v_i$. Thus, a computational measurement removes a given node from the cluster.

\subsubsection{Wire Shortening}

Sometimes we may also wish to remove a given vertex but still preserve the connections of its neighbors. This transformation is useful, for example, when we wish to shorten linear graph states. Consider $\hat{p}$-measurements on the two inner nodes of the four-node linear graph, which has nullifier basis
\begin{equation}
\{\hat{p}_1 - \hat{q}_2,\hat{p}_2 - \hat{q}_1 - \hat{q}_3,\hat{p}_3 - \hat{q}_2 - \hat{q}_4,\hat{p}_4 - \hat{q}_3\}.
\end{equation}
Since we will be making measurements of~$\hat{p}_2$ and~$\hat{p}_3$, we want a new basis in which only one nullifier fails to commute with~$\hat{p}_2$ and only one other fails to commute with~$\hat{p}_3$. We construct this basis out of linear combinations of the original basis elements, resulting in
\begin{equation}
\{\hat{p}_1 - \hat{q}_2,\hat{p}_2 - \hat{p}_4 - \hat{q}_1,- \hat{p}_1 + \hat{p}_3 - \hat{q}_4,\hat{p}_4 - \hat{q}_3\}
\end{equation}
Measurements of $\hat{p}_2$ and $\hat{p}_3$, with outcomes $m_2$ and $m_3$, respectively, collapse the cluster into a new graph state with nullifiers $\{m_2 - \hat{p}_4 - \hat{q}_1, m_3 -\hat{p}_1 - \hat{q}_4\}$. This is equivalent to the graph state of the two-qumode cluster, modulo a known quadrature displacement and a reflection in phase space about one of the nodes. Thus, measurements in the momentum basis allow us to effectively ``shorten'' linear graph states.

\begin{figure}[tb]
\centering
\includegraphics[width=\columnwidth]{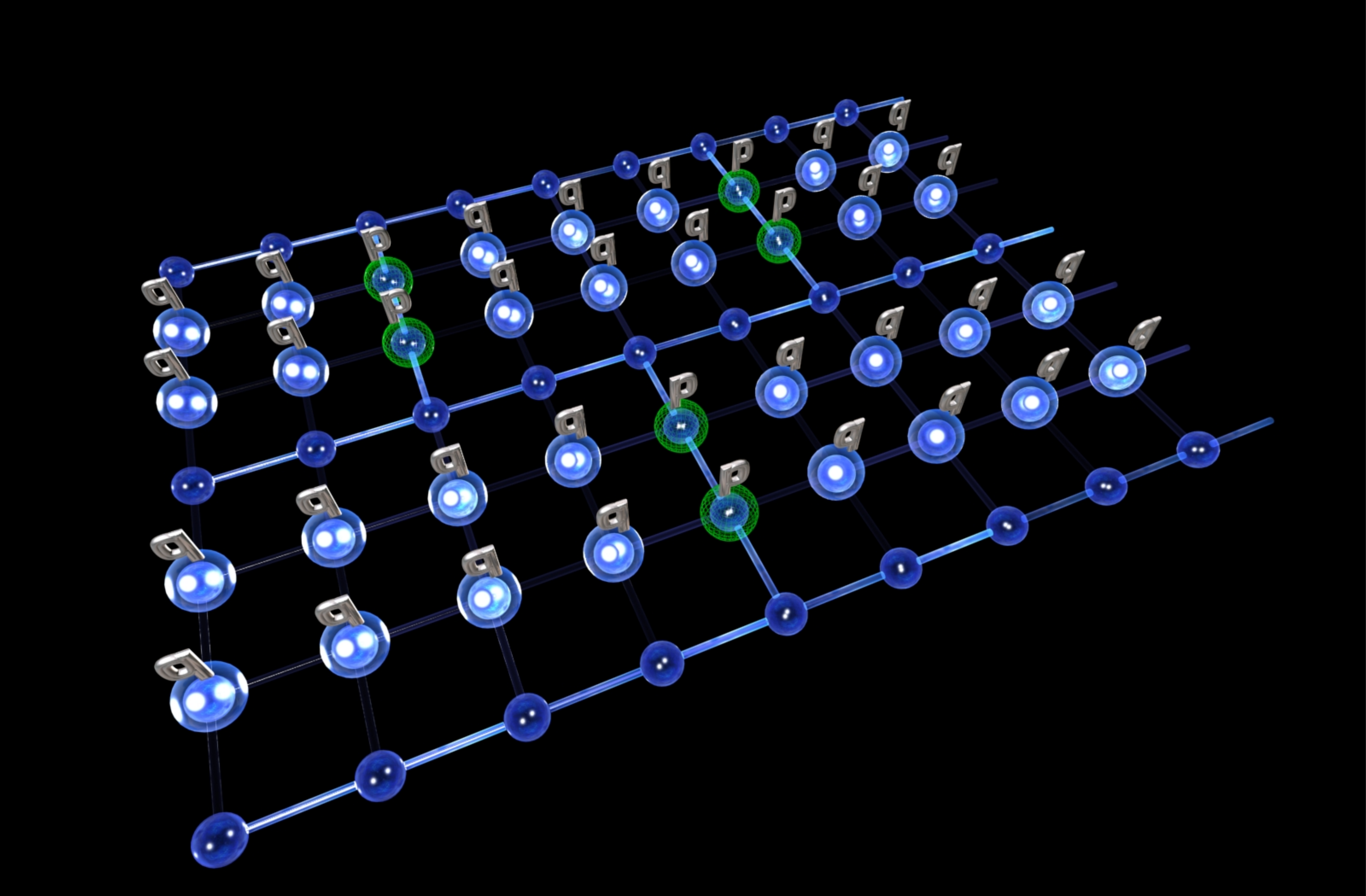}
\caption{\label{fig:universal_cluster} Any CV graph state may be generated by appropriate single-mode measurements. Computational-basis measurements (blue orbs) remove unwanted nodes. Momentum-basis measurements (green orbs) are then employed to shorten the ``wires'' within the cluster.}
\end{figure}

\subsection{The Universal Resource State}

Recall that any quantum circuit may be implemented by a sequence of measurements on a specifically tailored brickwork state (See Fig.~\ref{fig:circuit_mapping}(b)). The graph for such states is always a subgraph of a sufficiently large two-dimensional square lattice (see Fig.~\ref{fig:graph_states}(b)). The two transformations outlined above allow us to carve out an appropriate graph state for simulating any given circuit (Fig.~\ref{fig:universal_cluster}). The graph state that corresponds to a planar square lattice is thus a resource for universal CV quantum computation and is therefore a CV cluster state.

In practice, of course, lattices are always of finite size, just as are all quantum circuits. Therefore the complexity of the quantum computation one wishes to perform is constrained by the size of the original resource state. Since the size of the required cluster grows linearly with the number of fundamental one- and two-qumode gates and also grows linearly with the number of qumodes, CV cluster-state computation is efficient. As in the case of qubits, any algorithm of polynomial gate complexity can also be implemented by a resource state of polynomial size.
%

\section{The Effects of Finite Squeezing}

The ideal framework of CV quantum computation involves the use of momentum eigenstates. Such states cannot be normalized and are thus an idealized abstraction.  Any practical implementation must necessarily approximate these states.  One way to do so is by replacing each zero-momentum eigenstate with a vacuum state that has been finitely squeezed in the momentum quadrature.  In this section, we detail the resulting distortions imposed on any quantum computation that uses cluster states made from these approximate states.

Suppose we use states of finite squeezing, i.e.,\ $S(s)\ket{0}$ (where $\ket{0}$ represents the vacuum) for some large $s$, in place of momentum-quadrature eigenstates.  The resulting graph state obtained will not be ideal. Formally, we say that the resulting graph state $\ket{\phi(s)}$ is of accuracy~$s$. Such states are generalized CV cluster states, and Eq.~\eqref{eqn:general_cluster} allows us to write down their Wigner representation:
\begin{equation}
W(\mathbf{q},\mathbf{p}) = \prod_{j=1}^n G_{s}(q_j)G_{1/s}(H_j),
\end{equation}
where $G_{s}(q) = (\pi s^2)^{-1/2} \exp(-q^2/s^2)$ represents a Gaussian distribution with variance $s^2/2$, and $\{H_j\}$ are the nullifiers of $\ket{\phi(s)}$ from Eq.~\eqref{eqn:nullifers}.  Thus, these generalized CV cluster states are Gaussian states.  Observe that $G_{s}$ converges to a uniform distribution and $G_{1/s}$ to a $\delta$-peaked distribution in the limit of large $s$, in agreement with the Wigner function for an ideal graph state, Eq.~\eqref{eqn:wigner}.

To analyze the resulting distortions, we first consider the special case of simple state teleportation,
where only $\hat{p}$-measurements are made and the input state propagates through the cluster without any intended
manipulation. This result may then be extended to arbitrary measurements and the implementation of universal gates.

\subsection{Distortions in State Propagation}
 Consider the state resulting from a  $\hat{p}$-measurement on a $2$-qumode cluster, with input state $\ket{\phi}$ specified by the Wigner function~$W_{\text{in}}(q,p)$. In terms of the circuit model, this is represented by
\begin{equation}
\label{circ:distortedteleport}
\Qcircuit @C=1em @R=1em {
    \lstick{\ket{\phi}} & \ctrl{1} &  \qw & \measureD{\hat{p}} & \rstick{m} \cw
     \\
    \lstick{S(s)\ket{0}} & \ctrl{0} & \qw & \rstick{X(m) F \ket{\phi}'} \qw
}
\end{equation}
where $\ket{\phi}'$ is a distorted version of $\ket{\phi}$, which will be specified below.  We can analyze this circuit as follows.  After entangling the input, the state of the system is given by
\begin{equation}
W_{\text{in}}(q_1,p_1 - q_2)G_{s}(q_2)G_{1/s}(p_2 - q_1)
\end{equation}
A $\hat{p}_1$-measurement with outcome~$m$ yields the output state
\begin{align} \nonumber
P(m) W_{\text{out}}(q,p) &= G_{s}(q) \int d\tau\, W_{\text{in}}(\tau,m - q)G_{1/s}(p - \tau) , \\
&= G_s(q) \left[(W_{\text{in}} \ast_1 G_{1/s})(p,m-q) \right],
\end{align}
where $\ast_1$ denotes a convolution with respect to the first argument of $W$, and $P(m)$~is the probability of measurement outcome~$m$.  $P(m)$ multiplies the resulting, normalized pure state~$W_{\text{out}}$ to give the actual expression on the right-hand side.

What we would like to know from this toy example is how the imperfect squeezing affects the encoded state under the cluster-state implementation of the identity gate.  A good way to see this effect is to undo the unitary correction~$X(m)F$ and compare the result~$W_{\text{in}}'$ to the original input state~$W_{\text{in}}$:
\begin{align}
	P(m) W_{\text{in}}' = G_s(m-p) \left[(W_{\text{in}} \ast_1 G_{1/s})(q,p) \right].  \label{eqn:distortion}
\end{align}
The Wigner function~$W_{\text{in}}'(q,p)$ corresponds to $\ket{\phi}'$ in Circuit~\eqref{circ:distortedteleport}.  This means that with respect to the quantum information to be teleported, the Gaussian envelope is dependent on the measurement outcome~$m$.  Some values of~$m$ will result in an envelope that overlaps the (nonnegligible) support of~$W_{\text{in}}$, while other more extreme values of~$m$ will result in a strongly shifted envelope that cuts off large portions of the support of~$W_{\text{in}}$.  Thus, the actual success of any instance of teleportation depends strongly on the measurement outcome~$m$.

On the other hand, we can instead talk about the average state (a mixed state) that results from teleportation when we average over all possible measurement results~$m$ using their corresponding probabilities~$P(m)$.  This state is easily calculated using Eq.~\eqref{eqn:distortion}:
\begin{align}
	W_{\text{avg}} = \avg{W_{\text{in}}'} = \int dm\, P(m) W_{\text{in}}' = W_{\text{in}} \ast_1 G_{1/s}.
\end{align}
Thus, the average effect on the quantum information due to teleportation using finitely-squeezed resources is just the addition of
a variance of $1/(2 s^2)$ noise units on the $\hat{q}$-quadrature.  %
%
%
%
%
Repeated application gives us the resulting average distortion when a chain of $\hat{p}$-measurements is used to teleport an initial state $W_{\text{in}}$ down a linear cluster:
\begin{equation}
W_{\text{avg}} =  W_{\text{in}} \ast_1 G_{1/s} \ast_2 G_{1/s} \ast_1 G_{1/s} \ast_2 \dotsm
\end{equation}
In summary, when propagating quantum information through a chain of finite accuracy $s$, in every single shot,
pure, conditional output states are created with Gaussian envelopes applied to the input state in alternating quadratures and with the measurement results~$\{m_i\}$ determining their respective centers. More typical, when CV quantum information is teleported through a chain of
finite accuracy $s$, on average, $1/(2 s^2)$ units of noise are added alternately between the two quadratures resulting, in general, in a mixed output state. Whether a single-shot or an average picture is applicable depends on the actual experimental implementation and the encoding of the signal states.

\subsection{Distortions in Universal Gate Teleportation}

The distortions derived above, caused by finite squeezing, apply to all single-qumode measurements. To see this, consider the application of an arbitrary single-qumode unitary~$D$, diagonal in the computational basis, by measuring in the $D^\dag \hat{p} D$ basis. Since $D$ and $\CZ$ commute, this is equivalent to standard teleportation with input $D \ket{\phi}$, i.e.,
\[
\Qcircuit @C=1em @R=1em {
    \lstick{D\ket{\phi}} & \ctrl{1} &  \qw & \measureD{\hat{p}} & \rstick{m} \cw
     \\
    \lstick{S(s)\ket{0}} & \ctrl{0} & \qw & \rstick{X(m) F (D \ket{\phi})'} \qw
}
\]
Thus, the resulting output state is again the expected output state in the limit of ideal graph states, subjected to the distortion given by Eq.~\eqref{eqn:distortion}.  Therefore, the use of finite squeezing results universally in the addition of Gaussian noise that ``blurs out'' the details in the momentum and position quadratures, alternating between them at each step. The magnitude of this noise is inversely proportional to the accuracy of the cluster and grows linearly with the length of the cluster. This noise can potentially be reduced by the use of redundant rails~\cite{vanLoock2007}. However, such redundant, multiple-rail encoding requires a larger amount of squeezing resources for creating
the corresponding graph state. We will get back to this point in the following sections on optical cluster-state generation and computation.

\section{Optical Cluster States}

The optical implementation of CV cluster states holds particular promise and features a number of advantages over its discrete-variable counterpart~\cite{Nielsen2004,Browne2005}. With optical qubit cluster states, the entangling operation that is used to generate a cluster state is highly nonlinear, and its proposed implementations are nondeterministic. This results in significant overhead and presents an impediment to the generation of large-scale clusters. 

While challenging nonlinear operations are still required for universal quantum computation, the generation of CV cluster states with current technology is entirely deterministic. In particular,
\begin{enumerate}
\item Any CV graph state can be prepared completely via the interaction of squeezed vacuum states through a network of linear optical elements. Not only do we avoid the need for nonlinear interactions, but online squeezing is also unnecessary.
\item Once the cluster state is prepared, any multi-qumode Gaussian operation may be implemented entirely by quadrature measurements (homodyne detection).
\item The addition of photon counting allows for universal quantum computation.
\end{enumerate}

Indeed, CV clusters of up to four qumodes have already been experimentally realized~\cite{Yonezawa2008,Yukawa2008,Su2007}.  In addition, recent results show that the network of linear optical elements may be eliminated entirely in favor of frequency-encoded qumodes and a single OPO~\cite{Menicucci2007,Menicucci2008,Flammia2008a}.  Such a method would be able to create a CV cluster state in just one step and in a single beam of light.  Some such proposals also have significant scaling potential~\cite{Menicucci2008,Flammia2008a}.  What follows, however, will focus on the method described in Item~1 and discussed in detail in Ref.~\cite{vanLoock2007}.  Item~2 suggests that once such clusters are available, they can be immediately tested by implementing protocols involving information distribution and other Gaussian operations.  For example, this result immediately offers an experimentally viable method to use offline squeezed resources to perform squeezing operations online;
such an online CV gate operation using offline CV resource states can then be not only ``universally'' applied to arbitrary optical signal states~\cite{Yoshikawa2007,Yoshikawa2008,Filip2005}, but would also no longer require adjustment of the offline resources to achieve
different squeezing gates, as the CV cluster states provide a universal resource for Gaussian computation together with homodyne detectors~\cite{vanLoockJOSA2007}. Finally, while accurate photon counting remains experimentally challenging, Item~3 implies that universal quantum computation is nevertheless possible.


\subsection{Cluster State Generation}

The naive canonical method to generate a given CV cluster state would be to apply the theoretical definition directly, i.e., apply $\CZ$ interactions to a collection of squeezed states. While this method is conceptually simple, it is not very practical. The $\CZ$ operation does not conserve photon number, and requires the use of two single-mode online squeezers~\cite{Braunstein2005}.

In a more practical approach, in Ref.~\cite{vanLoock2007}, it was shown that online squeezers are not needed at all.
Any desired CV graph state of accuracy $s$ is a pure multi-mode Gaussian state, and hence the only necessary online components are passive linear optics~\cite{Braunstein2005,vanLoock2007}. To make this precise, consider the generation of a graph state $\ket{\phi(s)}$ corresponding to some graph $G$. Recall that $\ket{\phi(s)}$ is defined by application of an appropriate sequence of $\CZ$ gates to a collection of squeezed states $S(s)\ket{0}$.

\begin{figure}[tb]
\includegraphics[width=\columnwidth]{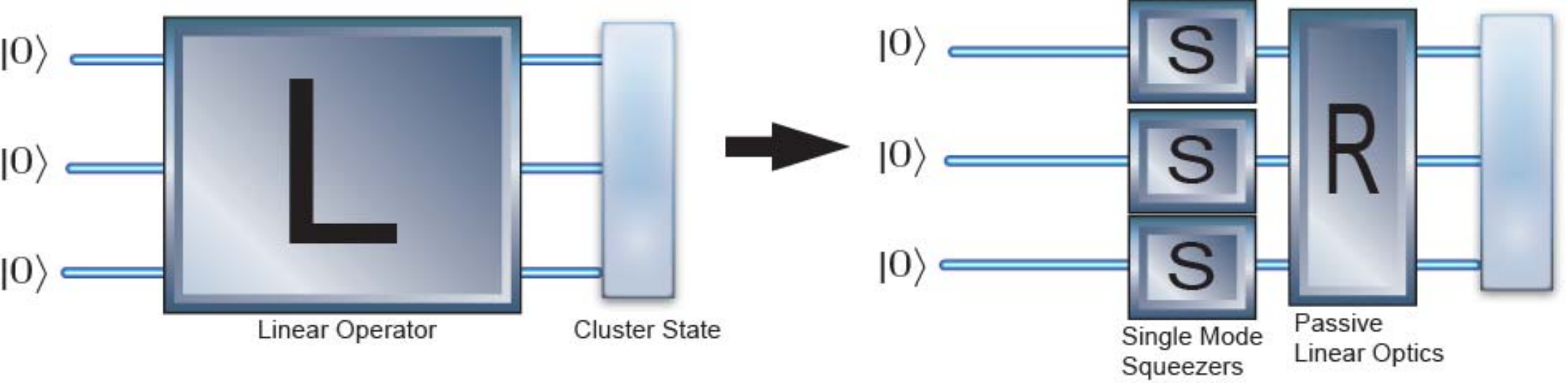}
\caption{\label{fig:cluster_prepare} A CV graph state can be generated by a Gaussian unitary acting on a collection of vacuum states~\cite{vanLoock2007}.  The Heisenberg action of this Gaussian is given by a symplectic linear operator~$L$ acting on the phase space of quadrature operators~$(\hat{q}_1,\hat{q}_2,\ldots,\hat{p}_1,\hat{p}_2,\ldots)$. This linear action can always be decomposed into the passing of offline squeezed states through a network of passive linear optical elements~\cite{Braunstein2005}.}
\end{figure}

The sequence of Gaussian transformations that take a collection of vacuum states to $\ket{\phi(s)}$ is represented succinctly in the Heisenberg picture.  Let $\hat{\mathbf{v}}$ denote the vector of quadrature operators:~$\hat{\mathbf{v}} = (\hat{q}_1,\hat{q}_2,\ldots,\hat{p}_1,\hat{p}_2,\ldots)$.  In the Heisenberg picture, the quadratures are transformed according to:
\begin{equation}
\label{eq:vtransform}
\hat{\mathbf{v}} \rightarrow \mathcal{C} \mathcal{S}(s) \hat{\mathbf{v}} = \mathcal{M}(s) \hat{\mathbf{v}},
\end{equation}
where $\mathcal{S}(s)$ represents the squeezing of each vacuum mode to form the state~$S(s)\ket{0}$, and $\mathcal{C}$ represents application of $\CZ$ operations in accordance with the desired graph. Mathematically, these operations can be defined by their action on the quadrature operators:
\begin{align}\nonumber
\mathcal{S}(s) \hat{q}_i &= s \hat{q}_i,  & \mathcal{C} \hat{q}_i &= \hat{q}_i,
\\
\mathcal{S}(s) \hat{p}_i &= \hat{p}_i/s, & \mathcal{C} \hat{p}_i &= \hat{p}_i + \sum_{(v_j,v_i) \in E} \hat{q}_j.
\end{align}
Concatenation of the two operations gives an explicit form for $\mathcal{M}(s)$. We refer to $\mathcal{M}(s)$ as the generation matrix for $\ket{\phi(s)}$, which defines the Gaussian operation that generates $\ket{\phi(s)}$ {\it from the vacuum}. The singular value decomposition for this matrix then provides an explicit recipe for how it may be generated using only linear optics and offline squeezing~\cite{Braunstein2005} (See Fig. \ref{fig:cluster_prepare}). We refer to this method as the decompositional method~\footnote{In Ref.~\cite{vanLoock2007}, the term ``canonical'' cluster states was reserved for those states that are obtained by directly applying a network of $\CZ$ gates onto momentum-squeezed states. These canonical states can then also be created, after Bloch-Messiah decomposition, with offline squeezing and linear optics. In this sense, the decompositional
scheme is equivalent to the ``canonical'' scheme based on online $\CZ$ gates. In addition to those schemes resulting in the canonical cluster states, in Ref.~\cite{vanLoock2007}, an alternative protocol was derived (independent of the Bloch-Messiah reduction), where offline squeezed states are sent through passive linear optics under the constraint that the outgoing multi-mode state satisfies the quadrature nullifier conditions in the limit of infinite squeezing.
In this case, for finite squeezing, there exist output states different from the canonical cluster states. This larger family of cluster states
was referred to as cluster-type states, including many ``non-canonical'' cluster states, nonetheless satisfying the nullifier conditions in the limit of infinite squeezing. With regard to experiments, an important feature of these generalized states is that the anti-squeezing components are
suppressed by construction~\cite{Yukawa2008}.
In the present paper, in order to make a comparison between the scheme based on Bloch-Messiah reduction and that using direct $\CZ$ gates, the former is here referred to as the ``decompositional method'', while the latter shall be named the ``canonical method''. Later, in order to make this comparison, no distinction will be made between offline and online squeezing.}.

\subsubsection{A simple example}

We illustrate the basic principles of this process by considering the explicit generation of the two-mode graph state, with generation matrix
\begin{equation}
\mathcal{M}(s) = \left(\begin{array}{cccc}
s & 0 & 0 & 0   \\
0 & s & 0 & 0 \\
0 & s & s^{-1} & 0   \\
s & 0 & 0 & s^{-1}   \\
\end{array} \right).
\end{equation}
The four singular values of this matrix are given by $\lambda_1 = \lambda_2 = \lambda_{+}$, $\lambda_3 = \lambda_4 = \lambda_{-}$ where
\begin{equation}
\lambda_{\pm} = \frac{\sqrt{1 + 2s^4 \pm \sqrt{1 + 4 s^8}}}{\sqrt{2}s}.
\end{equation}
$\lambda_{\pm}$ specifies the amount of offline squeezing required to generate the two-mode graph state to an accuracy $s$ (notice that $\lambda_- = \lambda_+^{-1}$). That is, two squeezed states of magnitude $\lambda_+$ (i.e., $S(\lambda_+)\ket{0}$), together with passive linear optical elements, may be used to generate this simple cluster state. Since we have transferred the online squeezing of the $\CZ$ operation into extra squeezing during the preparation process, $\lambda_+$ is generally greater than $s$. To generate the two-mode cluster to the same accuracy, our initial resources must be squeezed to a greater extent. In the case of the two-qumode cluster,
\begin{equation}
\lambda_+ \sim \sqrt{2} s\,, \qquad s \gg 1.
\end{equation}
Thus in the usual case where high accuracy is desired, we need to begin with states with a factor $\sqrt{2}$ more squeezing to achieve a graph state of the same accuracy. This factor is known as the squeezing overhead.

Since online squeezing generally represents a much greater experimental challenge than its offline equivalent, the decompositional method has a clear advantage over the canonical approach~\cite{vanLoock2007}.  However, it is also fair to ask whether this method is superior in \emph{all} situations. To test this, we consider the limiting case where offline squeezing is assumed to be as costly as online squeezing. One answer to this (in the affirmative) was already given in Ref.~\cite{vanLoock2007}, wherein it is shown that extra local squeezing is required to obtain a canonically generated CV cluster state from an $N$-mode Gaussian state in standard form~\cite{Adesso2006}.  We will revisit this result from another angle.  In this case, the measure of the resource requirements is the total amount of squeezing required---whether online or offline---measured additively in units of~dB.

We consider the toy case of the two-mode cluster here and follow with the general case in the following section. To generate this state up to accuracy $s$, the canonical method has two actions where squeezers are required:
\begin{description}
\item[(a)] Squeezing two vacuum states to $S(s)\ket{0}$, which requires two squeezers $10 \log(s^2)$ dB each.
\item[(b)] Application of a single $\CZ$ gate (QND interaction), which requires two online squeezers of $4.18$~dB each~\cite{Braunstein2005}.
\end{description}
In contrast, the decompositional method requires squeezing in the following two steps:
\begin{description}
\item[(a)] Squeezing two vacuum states to $S(s)\ket{0}$, as before.
\item[(b)] Squeezing these two states further by $10 \mathrm{log}(2) \sim 3$ dB each in order to account for the required squeezing overhead of~$\sqrt{2}$.
\end{description}
Thus the decompositional method saves $(2~\text{modes}) \times (4.18~\text{dB/mode} - 3~\text{dB/mode}) = 2.36~\text{dB}$ of squeezing for all values of $s$. In the next section, we show that the superiority of this method extends to general cluster states.

\subsubsection{Resource Requirements for General Graph States}
The above example motivates a general question: how much offline squeezing is required to create a graph state to accuracy~$s$? In the case where $\CZ$ gates are directly applied, we would need $n$ squeezed states of magnitude~$s$. To generate such a state entirely by offline squeezing up to equal accuracy, the initial nodes would necessarily need to feature greater squeezing. The practicality of the decompositional method hinges on the size of this overhead.

In this section, we show that there exist classes of universal cluster states whose squeezing overhead per mode does not increase with the size of the cluster. In addition, as in the case of two-mode clusters, the decompositional method remains superior even if squeezing arbitrary states (online squeezing) was as easy as squeezing the vacuum (offline squeezing). This is facilitated by a succinct method that computes the offline squeezing required to generate any given graph state.
\begin{thrm}We may generate an $n$-mode graph state with graph $G$ to an accuracy~$s$ by passing $n$ squeezed vacuum states---$S(s_i)\ket{0}$, with $i = 1,\ldots,n$---through a network of linear optical gates. Let $\mathbf{A}$ be the adjacency matrix of $G$ \cite{West2000}. In the limit of large squeezing ($s \gg 1$), the level of squeezing required for each mode, $s_i$, depends linearly on $s$, such that
\begin{equation}
s_i = s \sqrt{1 + k^2_i},
\end{equation}
where $k_i$ are the singular values of $\mathbf{A}$. Thus, the squeezing overhead is bounded above by $\sqrt{1 + k^2_i}$.\end{thrm}

\begin{proof}Consider a particular graph state with generation matrix~$\mathcal{M}(s)$. When the $1/s$ terms can be neglected, the $n$ largest singular values of $\mathcal{M}(s)$ then coincide with the singular values of the block-column matrix
\begin{equation}
s
\left(\begin{array}{c}
\mathbf{I}_n \\
\mathbf{A} \\
\end{array} \right),
\end{equation}
where $\textbf{I}_n$ denotes the $n \times n$ identity matrix. The result follows. $\blacksquare$
\end{proof}

Thus, the squeezing overhead of a given graph state is determined completely by the structure of its underlying graph. This theorem allows us to compute the exact resources required to generate any graph state. In certain situations, it is sufficient to know an upper bound on the amount of squeezed resources required.

\begin{thrm}To generate an $n$-mode graph state with graph $G$ to an accuracy~$s$ ($s \gg 1$), the maximal amount of squeezing required for any individual resource mode is bounded above by $(Ks)$, where
\begin{equation}\label{eqn:sq_bound}
K = \sqrt{1 + \mathrm{maxdeg}^2(G)},
\end{equation}
and where $\mathrm{maxdeg}(G)$ denotes the maximum degree of $G$, i.e.,\ the maximum number of edges connected to a single vertex.
\end{thrm}

\begin{proof}Consider first the case where $G$ is an $m$-regular graph (i.e.,\ each of its vertices has degree~$m$). Let $\mathbf{A}$ be its adjacency matrix. Each row and column of  $\mathbf{A}$ then sums to $K$. Therefore $\mathbf{A}$ may be written as the sum of $m$ permutation matrices, $P_1,P_2,\ldots,P_m$. Noting that the largest singular value of a given matrix is its spectral norm (denoted~$\norm{\cdot}$), we have
\begin{equation}
k_i \leq \norm{\mathbf{A}} \leq \sum_{j=1}^m \norm{P_j} \leq m,
\end{equation}
in agreement with Eq.~(\ref{eqn:sq_bound}).  To generalize this result to an arbitrary graph state~$G'$ with adjacency matrix~$\mathbf{A}'$ and maximum degree~$m$, observe that any graph of maximum degree~$m$ may be obtained by removing edges from an $m$-regular graph~$G$.  The spectral norm of an adjacency matrix strictly decreases with the removal of an edge, so we have $\norm{\mathbf{A}'} < \norm{\mathbf{A}} \leq m$. $\blacksquare$ \end{proof}

This theorem immediately implies that the squeezing overhead for universal cluster states of any fixed accuracy is bounded. Such states have maximal degree of $4$, and hence feature a squeezing overhead of $\sqrt{17}$. Thus, to guarantee the generation of a universal cluster to accuracy $s$, one would need to (a) generate a lattice of optical modes, each of which is squeezed up to $10 \log(s^2)$ dB and (b) proceed to squeeze each mode by a further $12.31$ dB. Meanwhile, quantum wires would require a maximum overhead of $\sqrt{5}$, and hence an additional $6.99$ dB of squeezing.

To see that the decompositional method is superior to the canonical method,
recall that each $\CZ$~gate requires two $4.18$ dB squeezers. In the case of a universal square lattice, the ratio of edges to vertices is~2:1. Thus, while the decompositional approach requires at most an additional $12.31$~dB per vertex, the two $\CZ$~gates applied per vertex would cost $16.72$~dB (since each $\CZ$ gate requires two online squeezers of magnitude $4.18$ dB). For a square lattice of size $N \times N$, the decompositional method would save approximately $4.41 N^2~\text{dB}$ of squeezing. Thus, not only is it more practical to perform universal quantum computation through squeezed offline resources, but it also turns out to be more efficient, i.e., cheaper in terms of squeezing resources required.

A typical setup would involve the generation of $n$ squeezed optical modes. These are then passed through a network of linear optical gates, of which the resulting entangled beams formally encode the desired cluster. Since there exist cluster states that are universal, the setup of this generation process does not need to be altered for different algorithms and hence may potentially be mass produced. The resulting beams can then be measured to perform the desired quantum computation.

As a final remark in this section, we come back to the question whether redundant, multiple-rail encoding may suppress the accumulation of finite-squeezing errors in a cluster computation~\cite{vanLoock2007}. According to Eq.~(\ref{eqn:sq_bound}), the offline squeezing per mode for generating a multiple-rail graph of accuracy $s$ with $m$ rails~\cite{vanLoock2007} is bounded above by $(K s)$ with $K=\sqrt{1 + m^2}$.
Therefore the initial squeezing variances may be as small as $1/(2 s^2) \times 1/(1+m^2)$, so roughly $1/(2 s^2 m^2)$ for large $m$.
This lower bound converges to zero faster than the actual reduction of the excess noise in the cluster computation which scales as $1/(2 s^2 m)$.
If one has access to squeezing resources with variances of $1/(2 s^2 m^2)$, one may better use them directly without multiple-rail encoding~\cite{vanLoock2007}. However, note that also here no complete proof for the effective failure of a
decompositional multiple-rail scheme is given, as the above analysis only relies upon bounds.

\subsection{Optical Cluster-State Computation}

The measurement of optical modes completes the optical implementation of CV cluster-state computation. Recall from Section~\ref{sec:m_comp} that {\it any} multi-qumode operation may be implemented by measurements that generate (a)~the shearing transformation $e^{i s\hat{q}^2/2}$ and (b)~the cubic phase gate $e^{i s\hat{q}^3/3}$ (while a nontrivial two-mode gate only requires measuring the $\hat p$ quadratures on a two-dimensional cluster state).

%

The implementation of (a) is reasonably straightforward. The required measurement basis, $\hat{p} + s\hat{q}$, is a rotated quadrature basis. We can write it in the standard form $r [\sin(\theta) \hat{q} + \cos(\theta) \hat{p}] = r \hat{p}_\theta$, where
\begin{equation}
r = \sqrt{1 + s^2}, \qquad \tan(\theta) = s.
\end{equation}
Thus, the optical implementation involves measurement in the rotated quadrature basis~$\hat{p}_\theta$, followed by a rescaling of the result by a factor of $r$. Therefore, all multi-qumode Gaussian operations may be achieved via simple quadrature measurements on a sufficiently connected graph state. While such operations are insufficient for universal quantum computation, they allow for general graph transformations, and thus many optical experiments that test the foundations of cluster-state quantum computation can be performed using just quadrature measurements---i.e.,\ homodyne detection.

The physical implementation of the cubic phase gate is more challenging. Since the required Hamiltonian to be implemented is no longer quadratic in the quadrature operators, a nonlinear optical element is required~\cite{Lloyd1999}. Two separate strategies may be employed, involving either embedding the nonlinear resource within the cluster (making it a non-Gaussian state) or using photon counting on the modes of an existing (Gaussian) graph state.

\subsubsection{Quantum Computation by Nonlinear Resources}

In the standard model of cluster-state computation, all qumodes are initialized in the state $\ket{0}_p$ prior to the entangling operation. The quantum computation is entirely dictated by our choice of measurement basis $D^\dag \hat{p} D$. However, it is also possible to encode the computation-gates~$D$ within our initial resource. To see this, note that since $D$ commutes with $\CZ$, the circuit
\[
\Qcircuit @C=1em @R=1em {
    \lstick{\ket{\phi}} & \ctrl{1} & \qw & \qw & \measureD{\hat p} & \rstick{m_1} \cw
     \\
    \lstick{\ket{0}_p} & \ctrl{0} & \ctrl{1} & \qw & \measureD{D^\dag \hat{p} D} & \rstick{m_2} \cw
    \\
    \lstick{\ket{0}_p} & \qw & \ctrl{0} & \qw & \qw\\
}
\]
is operationally equivalent to
\begin{eqnarray}\label{eqn:resource_circuit}
\Qcircuit @C=1em @R=1em {
    \lstick{\ket{\phi}} & \ctrl{1} & \qw & \qw & \measureD{\hat{p}} & \rstick{m_1} \cw
     \\
    \lstick{D \ket{0}_p} & \ctrl{0} & \ctrl{1} & \qw & \measureD{\hat{p}} & \rstick{m_2} \cw
    \\
    \lstick{\ket{0}_p} & \qw & \ctrl{0} & \qw & \qw\\
}
\end{eqnarray}
Therefore, instead of measuring in the $D^\dag \hat{p} D$ basis, we may have used instead the resource state $D \ket{0}_p$ for creating the initial cluster. Since the cubic phase gate $D = e^{is\hat{q}^3/3}$ allows for universal quantum computation, this observation suggests that the cubic phase state, $e^{is\hat{q}^3/3} \ket{0}_p$, will have the same effect. One method to optically generate such states is given in Ref.~\cite{Gottesman2001}.

Should these states be available, they may be attached at set locations within a universal cluster state in place of the usual~$\ket{0}_p$.  The resulting {\it non-Gaussian} cluster would be an improved resource for universal quantum computation.  Any CV unitary may be implemented employing such clusters, {\it even when one is limited to quadrature measurements only}. Of course, generating non-Gaussian quantum states remains an experimental challenge, and the cubic phase state is no exception.

\subsubsection{Cluster-State Implementation of the Cubic Phase Gate}
The previous observation suggests that if there exists a method of generating the cubic phase state by viable single-qumode measurements on a standard cluster state, then cubic phase gates may be applied to arbitrary inputs.

One possible approach~\cite{Gottesman2001} involves photon counting on one arm of a displaced two-mode squeezed state. This procedure may be summarized by the following quantum circuit:
\begin{eqnarray}
\label{circ: GKP cubic phase state}
    \Qcircuit @C=1em @R=1.8em {
     \lstick{S(s^{-1}) \ket 0}  & \qw & \multigate{1}{B}  &  \qw  & \gate{Z(r)} & \measureD{\hat{n}} & \rstick{n} \cw \\
       \lstick{S(s) \ket 0} & \qw & \ghost{B} & \qw 
     &  \rstick{\approx e^{i\gamma(n)\hat q^3} \ket{0}_p} \qw
     }
\end{eqnarray}
This circuit entangles two highly squeezed states, $S(s) \ket {0}$ and $S(s^{-1}) \ket {0}$, with $s \gg 1$, via a standard beamsplitter interaction $B$. A large momentum displacement~$Z(r)$, with $r \gg s$, is applied to the resulting two-mode squeezed state. We then make a photon counting measurement $\hat{n}$ on the displaced mode, which \textit{approximately} collapses the unmeasured qumode into a cubic phase state~$e^{i\gamma(n) \hat q^3} \ket{0}_p$, dependent of the measurement result $n$ through $\gamma(n) = (6 \sqrt{2n + 1})^{-1}$.  This procedure is essentially a measurement-based quantum computation---it involves the application of suitable measurements on an entangled resource. Thus, we may recast it into the form of a standard cluster-state computation. The two-mode squeezed state generated coincides with a two-qumode cluster state modulo a local Fourier transform on one of the nodes.  These observations allow us to construct a circuit that is consistent with the cluster-state formalism, while also functionally equivalent to Circuit~\eqref{circ: GKP cubic phase state}, i.e.,
\begin{eqnarray}
\label{circuit: cv_phase_state}
    \Qcircuit @C=1em @R=1.8em {
     \lstick{S(s) \ket 0}  & \qw & \ctrl{1}  &  \qw  & \gate{X(r)} & \measureD{\hat{n}} & \rstick{n} \cw \\
       \lstick{S(s) \ket 0} & \qw & \ctrl{0} & \qw 
     &  \rstick{\approx e^{i\gamma(n)\hat q^3} \ket{0}_p} \qw
     }
\end{eqnarray}
In this circuit, the initial entangled resource is a standard two-qumode cluster arranged in a linear configuration. The quadratures of this state are rotated with respect to the preparation using a simple beamsplitter, so a position displacement is applied to the first qumode, followed by a photon counting measurement. Just as in Circuit~\eqref{circ: GKP cubic phase state}, the second qumode is then collapsed to an approximate measurement-dependent cubic phase state.

Should the above circuit be attached to the second qumode of Circuit (\ref{eqn:resource_circuit}), we may apply the cubic phase gate $e^{i\gamma(n) \hat{q}^3}$ to an arbitrary input state. A cubic phase gate $e^{i a \hat{q}^3}$, for any~$a$, may be decomposed into a combination of $e^{i\gamma(n) \hat{q}^3}$ and two squeezers that depend on both~$n$ and~$a$~\cite{Gottesman2001}:
\begin{equation}\nonumber
S^\dag(t(n)) e^{i\gamma(n) \hat{q}^3} S(t(n)) = e^{i a \hat{q}^3}, \qquad  t(n) = \left[a/\gamma(n)\right]^{1/3}.
\end{equation}
The addition of these squeezers to Circuit (\ref{eqn:resource_circuit}) gives a measurement-based scheme to implement $e^{i a \hat{q}^3}$ for any desired~$a$ (modulo measurement dependent shifts in phase space):
\begin{eqnarray}
\Qcircuit @C=.8em @R=1em {
    \lstick{\ket{\phi}} & \gate{S(t(n))} & \ctrl{1} & \qw & \measureD{\hat{p}} & \rstick{m_1} \cw
     \\
    \lstick{e^{i\gamma(n)\hat{q}^3}  \ket{0}_p} & \qw &\ctrl{0} & \ctrl{1} & \measureD{\hat{p}} & \rstick{m_2} \cw
    \\
    \lstick{\ket{0}_p} & \qw & \qw & \ctrl{0}  & \gate{S^\dag(t(n))}  & \rstick{e^{ia\hat{q}^3} \ket{0}_p} \qw
}
\end{eqnarray}
Since the squeezing operation $S(t(n))$ is Gaussian, it can always be implemented by a suitable sequence of quadrature measurements on a linear cluster state. Combining this circuit with Circuit~(\ref{circuit: cv_phase_state}) leads to a cluster-state implementation of the cubic phase gate (see Fig.~\ref{fig:cubic_cluster}).

Observe that since the squeezing strength~$t(n)$ is dependent on the outcome~$n$ of the photon counting measurement (in addition to the fixed parameter~$a$), implementation of the squeezers must be done \emph{after} the photon counting. As with non-Clifford group operations on qubit clusters, adaptive measurements are involved, and hence the order of the measurements now matters.

\begin{figure}[tb]
\centering
\includegraphics[width=\columnwidth]{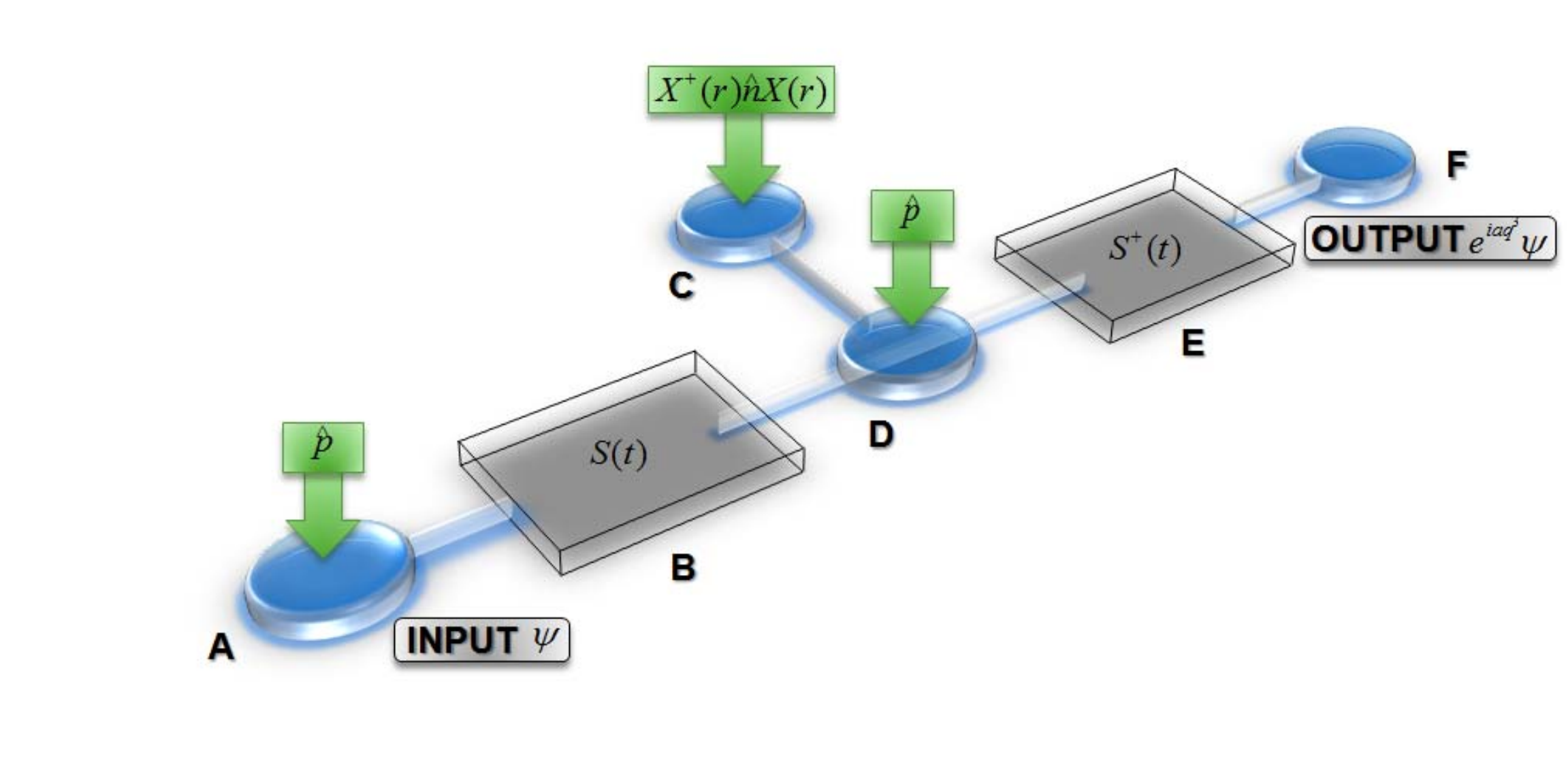}
\caption{\label{fig:cubic_cluster} A proposed cluster-state implementation of the cubic phase gate $e^{i a \hat{q}^3}$ applied to an arbitrary input state~$\ket{\phi}$. The displaced number-state measurement on node C implements Circuit~\eqref{circuit: cv_phase_state}, resulting in the generation of a cubic phase state at D. This nonlinear resource state, together with sub-clusters that generate the $n$-dependent squeezing corrections (B and~E), allows us to apply a cubic phase gate of any specified strength.}
\end{figure}


\section{Discussion and Conclusion}

Quantum information processing through the use of CV cluster states is a recent development in the field of quantum computation. While its basic principles have already been introduced in Ref.~\cite{Menicucci2006}, here we fleshed out the details of the protocol in a way that we hope will facilitate further understanding of CV cluster states, along with their potential optical implementation.

The Schr\"{o}dinger representation of CV graph states used in Ref.~\cite{Menicucci2006} becomes unwieldy as we explore graph states of nontrivial size. In such cases, Heisenberg nullifiers and Wigner functions are potentially useful tools. We showed how these representations may be utilized to derive rules of thumb on how graph states transform under measurements, as well as the existence of a CV cluster state that can be used as a resource for arbitrary CV operations.

The optical implementation of CV cluster states has also been further explored.  When the decompositional method was initially introduced in Ref.~\cite{vanLoock2007}, one of the primary concerns was that the price we pay for avoiding the need to perform online squeezing was an excessively large overhead in the extra offline squeezing required. Our results alleviate this concern.
We proved that the squeezing overhead per mode, when only offline squeezing is used for a given accuracy of the CV cluster,
does not increase with the size of a universal CV cluster state.
The upper bounds derived on the necessary amount of offline squeezing indicate that the decompositional approach has significant advantages over the direct approach through QND interactions~\cite{Menicucci2006}, even in cases where online squeezing is no more costly than its offline counterpart. While universal quantum computation using CV cluster states may be no less challenging than its qubit counterpart, the generation of CV clusters---either using the decompositional method~\cite{vanLoock2007} or through one-step generation using a single OPO~\cite{Menicucci2008}---is potentially more viable than in corresponding optical qubit schemes~\cite{Nielsen2004}.

To perform universal quantum computation, we adapted the experimental generation of the cubic phase gate as given in Ref.~\cite{Gottesman2001} to the cluster-state formalism. In addition, by extending the CV cluster-state framework to include the use of non-Gaussian resource states, we showed that possession of a suitable non-Gaussian state is sufficient for universal quantum computation within the cluster-state framework, even when only Gaussian operations, i.e., homodyne detections are employed during the cluster computation.
This leads to promising possibilities for universal CV cluster-state computation. One could, for example, envision that difficult nonlinear measurements are used to generate non-Gaussian resource states offline, which may then be distributed to consumers who are limited only to simpler measurements---i.e., quadrature homodyne detections. The consumers can nevertheless use the non-Gaussian states as resources for universal quantum computation.

In theory, all the ingredients for universal CV cluster-state computation have been developed. In particular, since the necessary squeezing resources
for creating a cluster of given accuracy have been shown to be independent of the size of a universal cluster state, scalability would only depend on the ability of suppressing the accumulation of errors at every measurement step during the cluster computation. At least for homodyne detections with near-unit efficiency, these errors mainly originate from the finite squeezing and they grow linearly with the length of the cluster. This could be compensated by increasing the accuracy of the cluster, hence making the squeezing per mode again dependent of the size of the cluster and the computation. Alternatively, some form of error correction may achieve full scalability in a strict sense, similar to fault tolerant schemes
for qubit quantum computation. To find an efficient method for CV quantum error correction, combined with a practical scheme for incorporating a non-Gaussian element into the Gaussian CV cluster-state framework, remains the main challenge to scalability and universality of the CV approach.
Nonetheless, our results here are an important step towards small-scale proof-of-principle demonstrations of CV cluster computation.

\acknowledgments

M.G., C.W. and T.C.R thank the support of the Australian Research Council (ARC).  N.C.M. was supported in part by the U.~S. Department of Defense and the U.~S. National Science Foundation.  Research at Perimeter Institute is supported by the Government of Canada through Industry Canada and by the Province of Ontario through the Ministry of Research \& Innovation. P.v.L. acknowledges support from the Emmy Noether programme of the DFG in Germany.


\begin{appendix}

\section{The Nullifier Formalism for Quadrature Measurements}\label{sec:nullifier}

Graph state transformations through quadrature measurements have an efficient nullifier description. Consider a measurement $\hat{p}_i$ on a graph state $\ket{\phi}$ nullified by a vector space with basis elements $\{H_1,H_2,\ldots, H_n\}$. There are only two distinct cases:
\begin{enumerate}
\item $\hat{p}_i$ commutes with all basis elements~$H_j$.  Then, $\ket{\phi}$ must be an eigenstate of $\hat{p}_i$ with some eigenvalue $m_i$.
\item There exists exactly one basis element that does not commute with $\hat{p}_i$; call this element $H_i$.
\end{enumerate}
This is because if there were to exist two or more basis elements that do not commute with $\hat{p}_i$, say $H_i$ and $H_j$, then there exists a constant $k$ such that $[\hat{p}_i,H_i + k H_j] = 0$. Thus, it is always possible to construct a basis such that only one element does not commute with $\hat{p}_i$, with $[\hat{p}_i,H_i] = -i$.

In case (1), $H_k \hat{p}_i \ket{\phi} = \hat{p}_i H_k \ket{\phi} = 0$ for each nullifier, and thus $\hat{p}_i \ket{\phi} = m_i \ket{\phi}$ for some $m_i$. Hence $(\hat{p}_i - m_i) \ket{\phi}$ must be a nullifier of $\ket{\phi}$ for some $m_i$. The measurement of $\hat{p}_i$ yields $m_i$, and the state remain undisturbed.

In case (2), we write $H_i$ as $\hat{q}_i + \sum_j c_j \hat{p}_j + \sum_{j \neq i} d_j \hat{q}_j + c_0$ for some constants $c_j$ and $d_j$. Let the measurement result be $m_i$.  Then, the transformed nullifier algebra is obtained by replacing $H_i$ with $\hat{p}_i - m_i$. Since the $i^{\text{th}}$ mode is now disentangled from the cluster and no longer interesting, we may discard it by choosing a basis such that all but one of the elements (namely $\hat{p}_i - m_i$) acts as the identity on the $i^{\text{th}}$ mode. Measurement in the computational basis~$\hat{q}$ can be analyzed analogously. More general quadrature measurements of the form $\hat{p}_{s\hat{q}^2/2} = \hat{p} + s \hat{q}$ can be treated in this formalism by application of the unitary $e^{is\hat{q}^2/2}$, followed by a standard momentum measurement.

We demonstrate this formalism on the case where $\hat{p}$-measurements are made on the first two qumodes (with measurement results $m_1$ and $m_2$) of a linear three-qumode cluster, defined by the three nullifiers
\begin{align}
\{H_1,H_2,H_3\} = \{\hat{p}_1 - \hat{q}_2, \hat{p}_2 - \hat{q}_1 - \hat{q}_3, \hat{p}_3 - \hat{q}_2\}.
\end{align}
$\hat{p}_1$ commutes with all nullifiers except $H_2$, so we replace $H_2$ with $\hat{p}_1 - m_1$, giving the nullifiers $\{\hat{p}_1 - \hat{q}_2, \hat{p}_1 - m, \hat{p}_3 - \hat{q}_2\}$, which defines the same state as $\{\hat{q}_2 - m_1, \hat{p}_3 - m_1\}$ after discarding the measured mode.

Repeating this procedure for the measurement of $\hat{p}_2$ results in the nullifier of the output state being~$\hat{p}_3 - m_1$.  Thus the remaining unmeasured node is in the state $\ket{m_1}_p$, in agreement with Eq.~\eqref{eqn:single_circuit}.


\end{appendix}


\end{document}